\documentclass[numberwithinsect,a4paper,cleveref,autoref]{no-lipics-v2019} %Lipics 2019

\allowdisplaybreaks[1]

\usepackage{chessfss}
\usepackage{accents}

\newtheorem{fact}[theorem]{Fact}
\newtheorem{observation}[theorem]{Observation}

\usepackage{amssymb}
\usepackage{bm}
\usepackage{mathtools}
\newcommand{\enc}{\mathsf{enc}}
\newcommand{\cc}[1]{\mathsf{#1}}

\newcommand{\supp}{\ensuremath{\mathsf{supp}}}
\newcommand{\tw}{\ensuremath{\mathsf{tw}}}

\newcommand{\W}[1]{\ensuremath{\mathsf{W[#1]}}}
\newcommand{\A}[1]{\ensuremath{\mathsf{A[#1]}}}

\newcommand{\homs}[2]{\mbox{\ensuremath{\mathsf{Hom}(#1 \to #2)}}}

\newcommand{\embs}[2]{\mbox{\ensuremath{\mathsf{Emb}(#1 \to #2)}}}

\newcommand{\subs}[2]{\mbox{\ensuremath{\mathsf{Sub}(#1 \to #2)}}}

\newcommand{\strembs}[2]{\mbox{\ensuremath{\mathsf{StrEmb}(#1 \to #2)}}}

\newcommand{\indsubs}[2]{\mbox{\ensuremath{\mathsf{IndSub}(#1 \to #2)}}}
\newcommand{\auts}[1]{\ensuremath{\mathsf{Aut}(#1)}}
\newcommand{\cphoms}[2]{\ensuremath{\mathsf{cp}\text{-}\mathsf{Hom}}(#1 \to #2)}

\newcommand{\cfhoms}[2]{\ensuremath{\mathsf{cf}\text{-}\mathsf{Hom}}(#1 \to #2)}

\newcommand{\clique}{\ensuremath{\textsc{Clique}}}

\newcommand{\homsprob}[2]{\mbox{\ensuremath{\textsc{Hom}(#1 \to #2)}}}
\newcommand{\cphomsprob}[2]{\mbox{\ensuremath{\textsc{cp-Hom}(#1 \to #2)}}}
\newcommand{\subsprob}[2]{\mbox{\ensuremath{\textsc{Sub}(#1 \to #2)}}}

\newcommand{\fptred}{\ensuremath{\leq^{\mathrm{fpt}}_{\mathrm{T}}}}
\newcommand{\interred}{\ensuremath{\equiv^{\mathrm{fpt}}_{\mathrm{T}}}}
\newcommand{\parsired}{\ensuremath{\leq^{\mathrm{fpt}}}}
\newcommand{\wparsired}{\ensuremath{\leq^{\mathrm{w\text{-}fpt}}}}

\newcommand{\lovasz}{Lov{\'{a}}sz}
\newcommand{\nesetril}{Ne{\v{s}}et{\v{r}}il}

\renewcommand{\ast}{\ensuremath{\!\raisebox{1ex}{${*}$}\!}}
\newcommand{\K}{\ensuremath{\mathrm{K}}}
\newcommand{\xcrown}[1]{\accentset{\symking}{#1}}

\title{Counting and Finding Homomorphisms is Universal for Parameterized Complexity Theory}
\titlerunning{Counting Homomorphisms is Universal for Parameterized Complexity Theory}

\author{Marc Roth}
{Merton College, University of Oxford, United Kingdom}
{marc.roth@merton.ox.ac.uk}
{https://orcid.org/0000-0003-3159-9418}
{}

\author{Philip Wellnitz}
{Max Planck Institute for Informatics, Saarland Informatics Campus (SIC), Saarbrücken, Germany
}
{wellnitz@mpi-inf.mpg.de}
{https://orcid.org/0000-0002-6482-8478}
{}

\authorrunning{M. Roth and P. Wellnitz}

\Copyright{Marc Roth and Philip Wellnitz}



\acknowledgements{We thank Karl Bringmann and Holger Dell for fruitful discussions and valuable feedback on early drafts of this work.}

\nolinenumbers
\usetikzlibrary{kneser, conference}
\usetikzlibrary{arrows.meta}
\usetikzlibrary{shapes.misc, decorations.pathreplacing}

\begin{document}
\maketitle

\begin{abstract}
    Counting homomorphisms from a graph $H$ into another graph $G$ is a fundamental problem
    of (parameterized) counting complexity theory. In this work, we study the case
    where {\em both} graphs $H$ and $G$ stem from given classes of graphs: $H\in \mathcal{H}$
    and $G\in \mathcal{G}$.
    By this, we combine the \emph{structurally restricted} version of this problem
    (where the class $\mathcal{G}=\top$ is the set of all graphs), with the
    \emph{language-restricted} version (where the class $\mathcal{H}=\top$ is the set of all
    graphs).
    The structurally restricted version allows an exhaustive complexity classification
    for classes $\mathcal{H}$:
    Either we can count all homomorphisms in polynomial time (if the treewidth of $\mathcal{H}$
    is bounded), or the problem becomes $\#\W1$-hard [Dalmau, Jonsson, Th.Comp.Sci'04].
    In contrast, in this work, we show that the combined view most likely does not admit such
    a complexity dichotomy.

    Our main result is a construction based on Kneser graphs that associates every problem
    $\textsc{P}$ in $\#\W{1}$ with two classes of graphs $\mathcal{H}$ and $\mathcal{G}$ such that
    the problem $\textsc{P}$ is {\em equivalent} to the problem
    $\#\homsprob{\mathcal{H}}{\mathcal{G}}$ of counting homomorphisms
    from a graph in $\mathcal{H}$ to a graph in $\mathcal{G}$.
    In view of Ladner's seminal work on the existence of $\mathsf{NP}$-intermediate problems
    [J.ACM'75] and its adaptations to the parameterized setting, a classification of
    the class $\#\W1$ in fixed-parameter tractable and $\#\W1$-complete cases is unlikely.
    Hence, obtaining a complete classification for the problem
    $\#\homsprob{\mathcal{H}}{\mathcal{G}}$ seems unlikely.
    Further, our proofs easily adapt to $\W1$ and the problem of deciding whether
    a homomorphism between graphs exists.

    In search of complexity dichotomies, we hence turn to special graph classes.
    Those classes include line graphs, claw-free graphs, perfect graphs, and combinations thereof,
    and $F$-colorable graphs for fixed graphs $F$:
    If the class $\mathcal{G}$ is one of those classes
    and the class $\mathcal{H}$ is closed under taking minors, then we establish explicit criteria
    for the class $\mathcal{H}$ that partition the family of problems
    $\#\homsprob{\mathcal{H}}{\mathcal{G}}$ into polynomial-time solvable and $\#\W{1}$-hard
    cases. In particular, we can drop the condition of $\mathcal{H}$ being minor-closed for
    $F$-colorable graphs. As a consequence, we are able to lift the framework of graph motif
    parameters due to Curticapean, Dell and Marx [STOC'17] to $F$-colorable graphs and
    provide an exhaustive classification for the parameterized subgraph counting problem on
    $F$-colorable graphs. As a special case, we obtain an easy proof of the parameterized
    intractability result of the problem of counting $k$-matchings in bipartite graphs.
\end{abstract}

\clearpage

\section{Introduction}

\paragraph*{Homomorphisms between Graphs}

Given graphs $H$ and $G$, a fundamental question is whether (or how often) we can ``find'' the
graph $H$ in the graph $G$.
Depending on the application, different notions of ``finding'' the graph $H$ are studied:
In the classical {\em subgraph isomorphism} problem~\cite{Cook71,Ullmann76}
(also consult for instance \cite{AlonYZ95,NesetrilP85}), the goal is to search for subgraphs of $G$ that
are {\em isomorphic} to the graph $H$.
In contrast, allowing multiple vertices of the graph $H$ to be mapped
to the same vertex of $G$, requiring only edge relations to be preserved, we obtain
a relaxation of the subgraph isomorphism problem called the {\em (graph) homomorphism} problem.
The problem of finding (and by extension counting) homomorphisms in graphs
has a tight connection to problems related to conjunctive queries in data
bases~\cite{ChandraM77,DellRW19icalp}, as well as applications in for instance artificial
intelligence~\cite{FederV98}.
As it turns out, once we want to count all ``occurrences'' of the graph $H$ in $G$, understanding
the graph homomorphism problem is already enough to understand various other notions of
``finding'' graphs in other graphs:
As Curticapean, Dell, and Marx~\cite{CurticapeanDM17} proved, counting for instance isomorphic subgraphs
is the same as counting {\em linear combinations} of graph homomorphisms. Hence, in this work,
we focus on counting graph homomorphisms.\vspace*{1ex}

Formally, given two classes of graphs $\mathcal{H}$ and $\mathcal{G}$, the (decision)
problem $\homsprob{\mathcal{H}}{\mathcal{G}}$ is defined as follows: Given graphs
$H \in \mathcal{H}$ and $G \in \mathcal{G}$, decide whether there is a mapping $h: V(H) \to V(G)$
such that for any edge $uv$ in $V(H)$, the edge $h(u)h(v)$ exists in $V(G)$.
The problem of finding graph homomorphisms, also called $H$-colorings,
has been studied since the late 1970s and 1980s~\cite{GareyJS75,MaurerSW81,AlbertsonC85}.
In its most general form, where both classes $\mathcal{H}$ and $\mathcal{G}$ contain
all graphs (denoted by $\mathcal{H} = \mathcal{G} = \top$) the problem $\homsprob{\mathcal{H}}
{\mathcal{G}}$ is $\cc{NP}$-complete: Checking whether a graph $H$ admits a homomorphism to
the complete graph on $3$ vertices is equivalent to checking whether $H$ is 3-colorable,
a classical $\cc{NP}$-hard problem~(see for instance \cite{GareyJ79}).

Motivated by the hardness in the general case, special cases of the problem
$\homsprob{\mathcal{H}}{\mathcal{G}}$ have been studied:
The \emph{language-restricted} version of the problem $\homsprob{\mathcal{H}}{\mathcal{G}}$
assumes only the class $\mathcal{H}=\top$ to be the set of all graphs and restricts
the class $\mathcal{G}$. Note that the above example of checking whether a graph is
3-colorable falls into this framework (by setting $\mathcal{G}=\{K_3\}$), so the problem
$\homsprob{\top}{\mathcal{G}}$ is $\cc{NP}$-hard in general as well.
However, if the class $\mathcal{G}$ contains only {\em bipartite} graphs, the problem
$\homsprob{\mathcal{\top}}{\mathcal{G}}$ is solvable in polynomial time~\cite{HellN90}.
In fact, Hell and \nesetril~\cite{HellN90} also prove the following hardness result:
If the class $\mathcal{G}$ contains a non-bipartite graph, the problem
$\homsprob{\mathcal{\top}}{\mathcal{G}}$ is $\cc{NP}$-hard.
Together, this yields a {\em complexity dichotomy}: for any problem $\homsprob{\top}{\mathcal{G}}$,
we obtain its complexity just by looking at the class $\mathcal{G}$.
Based on this dichotomy, Feder and Vardi conjectured \cite{FederV98} that a similar dichotomy
is possible for constraint satisfaction problems; this Feder-Vardi-Conjecture was proved
recently by Bulatov~\cite{Bulatov17} and, independently, by Zhuk~\cite{Zhuk17}.\vspace*{1ex}

Having understood the problem $\homsprob{\top}{\mathcal{G}}$,
the focus shifted to understanding ``the other side'', that is
the case where the class $\mathcal{G}$ contains all graphs instead.
For this {\em structurally restricted} version of the problem
$\homsprob{\mathcal{H}}{\mathcal{G}}$, the success story continues:
As Grohe~\cite{Grohe07} proved, if every graph in the graph class $\mathcal{H}$
contains only graphs for which the so-called {\em treewidth} is small,\footnote{Strictly speaking,
the graphs in the class $\mathcal{H}$ need to be only {\em homomorphically equivalent} to
graphs with small treewidth.} then the problem
$\homsprob{\mathcal{H}}{\top}$ is solvable in polynomial-time again.
In all other cases, the problem $\homsprob{\mathcal{H}}{\top}$ is not solvable in
polynomial time (under reasonable assumptions from complexity theory).
In fact, Grohe's dichotomy is even stronger: It shows that from a {\em parametrized
complexity} view, the problem $\homsprob{\mathcal{H}}{\top}$ is either what is called
``fixed-parameter tractable'' (solvable in time $f(|V(H)|)\cdot |V(G)|^{O(1)}$ for graphs
$H\in \mathcal{H}$, $G\in\top$) or ``$\W1$-hard'' (essentially a parameterized
equivalent of $\cc{NP}$-hardness). We formalize these notions later; also consult
\cite{Niedermeier02,FlumG06,CyganFKLMPPS15} for an in-depth introduction to parameterized
complexity theory.

\paragraph*{The Doubly Restricted Version of Counting Homomorphisms}

A natural generalization of finding a solution to a problem is to {\em count} all solutions.
From an algorithmic point of view, counting all solutions may be way harder than finding
a solution: While finding a perfect matching in a graph has a classical polynomial-time
algorithm, {\em counting} all perfect matchings is known to be
$\cc{\#P}$-complete~\cite{Valiant79}.

Formally, for two classes of graphs $\mathcal{H}$ and $\mathcal{G}$, the counting
version of the homomorphism problem (denoted by $\#\homsprob{\mathcal{H}}{\mathcal{G}}$)
is defined as follows: Given graphs $H \in \mathcal{H}$ and $G \in \mathcal{G}$,
compute the number of (graph) homomorphism from the graph $H$ to the graph $G$.
Similarly to the decision realm, the language-restricted version $\#\homsprob{\top}{\mathcal{G}}$
has been studied in the context of the counting constraint satisfaction problem:
The dichotomy theorem of Dyer and Greenhill implies that the problem
$\#\homsprob{\top}{\mathcal{G}}$ is $\#\cc{P}$-complete if the class $\mathcal{G}$ contains a
graph with a connected component that is neither an isolated vertex with or without self-loop,
nor a complete graph with all self-loops, nor a complete bipartite graph without
self-loops~\cite{DyerG00}.
Otherwise, the problem $\#\homsprob{\top}{\mathcal{G}}$ is solvable in polynomial
time~(cf.\ \cite[Lemma~4.1]{DyerG00}). In a subsequent line of research, this classification was
lifted to general counting constraint satisfaction problems~\cite{Bulatov13,DyerR10,CaiC12}.

The structurally restricted version of the graph homomorphism problem
has been studied in the counting regime as well:
A counting analogue of Grohe's dichotomy was established by Dalmau and Jonsson~\cite{DalmauJ04}.
In \cite{DalmauJ04} they prove that the counting problem $\#\homsprob{\mathcal{H}}{\top}$ is
solvable in polynomial time if and only if there is a constant bound on the treewidth of the
graphs in the class $\mathcal{H}$; otherwise the problem $\#\homsprob{\mathcal{H}}{\top}$ is
complete for the class $\#\W{1}$ (where the class $\#\W1$ is the analogous class
to~$\W{1}$ for counting problems).

Initiated by the breakthrough result by Curticapean, Dell, and
Marx~\cite{CurticapeanDM17}, a line of research~\cite{Roth17,RothS18,DellRW19icalp} lifted the
dichotomy of Dalmau and Jonnson~\cite{DalmauJ04} to all parameterized counting problems that can
be expressed as linear combinations of homomorphisms, subsuming counting of subgraphs,
counting of induced subgraphs and even counting of answers to existential first-order queries.
This lifting technique is sometimes also called {\em complexity monotonicity}.

\paragraph*{Counting Homomorphisms is Universal for $\W1$ and $\#\W1$}

The previous results provide a surprisingly clean picture of the complexity landscape of the
problems of finding and counting graph homomorphisms for both, the language-restricted and the
structurally restricted version.
However, none of the previous results are applicable for the {\em doubly restricted} version:
Instead of restricting only $\mathcal{H}$ or $\mathcal{G}$,
we consider the problem $\homsprob{\mathcal{H}}{\mathcal{G}}$ where {\em both} classes
are fixed. This can be seen as a special case of both the structurally restricted version
and the language-restricted version. In particular, the known dichotomies translate only
for certain pairs of classes $\mathcal{H}, \mathcal{G}$, leaving a wide gap in the complexity
landscape to be explored. In particular, it is imaginable that for real-world instances,
{\em both} graphs $H$ and $G$ have a certain structure that can be exploited.
In fact, we show that the doubly restricted version can express {\em any} problem in $\W1$ and
$\#\W1$, respectively. Intuitively, this means that if we want to understand any problem $P$ in $\#\W1$,
we may instead consider an equivalent problem $\#\homsprob{\mathcal{H}_P}{\mathcal{G}_P}$. In
particular, any algorithm or hardness obtained for $\#\homsprob{\mathcal{H}_P}{\mathcal{G}_P}$
directly translates to the original problem $P$.
\begin{mtheorem}[Universality for $\W{1}$ and $\#\W1$]\label{ladner_intro}
    For any problem $P$ in $\W{1}$, there are classes~$\mathcal{H}=\mathcal{H}_P$ and
    $\mathcal{G} = \mathcal{G}_P$ such that \[P \interred \homsprob{\mathcal{H}}{\mathcal{G}},\]
    and for any problem $P'$ in $\#\W{1}$, there are classes~$\mathcal{H'} = \mathcal{H'}_P$
    and $\mathcal{G'} = \mathcal{G'}_P$ such that
    \[P' \interred \#\homsprob{\mathcal{H'}}{\mathcal{G'}},\]
    where $\interred$ denotes interreducibility (sometimes also called equivalence)
    with respect to parameterized Turing-reductions.\\
    The classes~$\mathcal{H}$ and $\mathcal{H'}$ are recursively enumerable and the classes
    $\mathcal{G}$ and $\mathcal{G'}$ are recursive.\lipicsEnd
\end{mtheorem}

\Cref{ladner_intro} in turn also makes a clear categorization of the problems
$\homsprob{\mathcal{H}}{\mathcal{G}}$ into ``easy'' (that is fixed-parameter tractable)
and ``hard'' (that is $\W1$-hard or $\#\W1$-hard) cases unlikely:
A general partition of the class $\W{1}$ in fixed-parameter tractable and $\W{1}$-complete
problems is very unlikely as indicated by Ladner's seminal result~\cite{Ladner75} and its
adaptation to the parameterized setting by Downey and Fellows~\cite{DowneyF92ladner}.
A similar reasoning applies to $\#\W{1}$.

Note that \cref{ladner_intro}, in particular its consequences for the absence of parameterized
dichotomies, are independent from the ``non-dichotomy'' results of~\cite{BodirskyG08}
and~\cite{ChenTW08}, which rule out a $\cc{P}$ vs.\ $\cc{NP}$/$\#\cc{P}$ dichotomy for the
structurally restricted versions: In \cite{BodirskyG08},  Bodirsky and Grohe prove
a $\cc{P}$ vs.\ $\cc{NP}$ non-dichotomy by a modification of Ladner's Theorem~\cite{Ladner75};
however, this has no direct implications from neither a parameterized complexity nor a counting
complexity point of view. Independently, in \cite{ChenTW08}, Chen, Thurley, and Weyer proved a
similar result also for the counting version and hence obtained a $\cc{P}$ vs. $\#\cc{P}$
non-dichotomy result; again, this has no direct implications for our setting.

\paragraph*{Dichotomies for $F$-Colorable Graphs and Kőnig Graphs}

Having established the doubly restricted version of the problem $\#\homsprob{\mathcal{H}}
{\mathcal{G}}$ as interesting in general, we proceed to demonstrate examples of both,
(1) how existing complexity dichotomies translate to the doubly restricted setting, as well as
(2) how we can exploit the existence of structure in both classes to obtain new complexity
dichotomies for certain pairs of graph classes.

Note that if we fix a graph class $\mathcal{G}$ for which the corresponding
language restricted problem $\#\homsprob{\top}{\mathcal{G}}$ is already ``easy'',
then the same is true for any graph class $\mathcal{H}$ and the problem
$\#\homsprob{\mathcal{H}}{\mathcal{G}}$.\footnote{For instance if $\mathcal{G}$ is the class of
all planar graphs, Eppstein~\cite{Eppstein99} gave an fixed-parameter tractable algorithm
for $\#\homsprob{\top}{\mathcal{G}}$; a similar result is known even for classes of bounded
local treewidth~\cite{FrickG01}. Hence, there are also fixed-parameter tractable algorithms
solving $\#\homsprob{\mathcal{H}}{\mathcal{G}}$ for any graph class~$\mathcal{H}$.\!}
While it may be possible to improve the running time of
known algorithms for special classes $\mathcal{H}$ in such a case, in this work we focus on
investigating classes $\mathcal{G}$ where the problem $\#\homsprob{\top}{\mathcal{G}}$
is hard. (Note further that a similar statement is true for classes $\mathcal{H}$ where
the structurally-restricted problem $\#\homsprob{\mathcal{H}}{\top}$ is ``easy''.)\vspace*{1ex}

As a first example how known dichotomies can be adapted to yield new results for the
doubly-restricted setting, we consider the case where the class $\mathcal{G} = \mathcal{G}_F$
is the set of all $F$-colorable graphs for some fixed graph $F$.\footnote{Observe that
containment in the class $\mathcal{G}_F$ is in general not solvable in polynomial time:
If $F$ is the triangle then $\mathcal{G}_F$, if considered as language, is the $3$-coloring
problem. For this reason, we model the problem $\#\homsprob{\mathcal{H}}{\mathcal{G}}$
as a \emph{parameterized promise} problem; the formal definition is given in \cref{sec:prelims}.}
For example, if $F$ is chosen to be the graph consisting of a single edge, then the problem
$\#\homsprob{\mathcal{H}}{\mathcal{G}_F}$ is the problem of counting homomorphisms from a graph
$H \in \mathcal{H}$ to a bipartite graph $G$.
As it turns out, it is possible to refine the dichotomy by Dalmau and Jonsson~\cite{DalmauJ04}
for the case $\mathcal{G}=\top$ to work for $F$-colorable graphs as well:
\begin{mtheorem}\label{thm:genhomdich_intro}
    Let $F$ denote a graph, and let $\mathcal{H}$ denote a recursively enumerable class of graphs.
    \begin{enumerate}[~1~]
        \item If the treewidth of the class $\mathcal{H}\cap\mathcal{G}_F$ is bounded,
            then the problem $\#\homsprob{\mathcal{H}}{\mathcal{G}_F}$ is polynomial-time solvable.
        \item Otherwise, the problem $\#\homsprob{\mathcal{H}}{\mathcal{G}_F}$ is
            $\#\W{1}$-complete.\lipicsEnd
    \end{enumerate}
\end{mtheorem}

While the proof \cref{thm:genhomdich_intro} is conceptually close to the proof in~\cite{DalmauJ04}, we can combine \cref{thm:genhomdich_intro} with
the aforementioned complexity monotonicty~\cite{CurticapeanDM17} to lift the result to the
realm of counting subgraphs:
\begin{mtheorem}[Intuitive version]\label{thm:subs_f_colored_intro}
    Let $F$ be a fixed graph and let $\mathcal{H}$ be a recursively enumerable class of graphs.
    Given a graph $H \in \mathcal{H}$ and a graph $G \in \mathcal{G}_F$, we wish to compute the
    number of subgraphs $\#\subs{H}{G}$ of $G$ that are isomorphic to $H$.
    \begin{enumerate}[~1~]
        \item If the matching number of $\mathcal{H}\cap \mathcal{G}_F$ is bounded then
            the problem $\#\subsprob{\mathcal{H}}{\mathcal{G}}$ is polynomial-time solvable.
        \item Otherwise, the problem $\#\subsprob{\mathcal{H}}{\mathcal{G}}$ is $\#\W{1}$-complete.
            \lipicsEnd
    \end{enumerate}

\end{mtheorem}
Note that \cref{thm:subs_f_colored_intro} subsumes a dichotomy for counting subgraphs
in bipartite graphs and, in particular, \cref{thm:subs_f_colored_intro} yields an alternative
and {\em easy} proof of $\#\W{1}$-hardness of counting $k$-matchings in bipartite
graphs~\cite{CurticapeanM14}.
Further, as an example of a new result which follows from \cref{thm:subs_f_colored_intro},
we obtain $\#\W{1}$-hardness for the problem of counting triangle packings in $3$-colorable graphs:
This problem asks, given parameter $k$ and a $3$-colorable graph $G$, to compute the number of
possibilities to embed $k$ vertex-disjoint triangles into $G$.\vspace*{1ex}

As an example for completely new insights gained in the doubly restricted setting, we consider
the cases where the class $\mathcal{G} = \mathcal{L}$ is the set of {\em line graphs} and
where the set $\mathcal{G} = \symking$ is the set of {\em Kőnig graphs}, respectively;
where a Kőnig graph is a line graph of a bipartite graph.\footnote{We chose this terminology
due to the fact that Kőnig's theorem states that line graphs of bipartite graphs are perfect
(see for instance\ \cite{ChudnovskyRST06}). The symbol $\symking$ is used since ``Kőnig'' is the German
word for ``King''.}
Kőnig graphs are of particular interest, as they are a subset of the well-studied classes of
perfect graphs~\cite{ChudnovskyRST06}, line graphs (of arbitrary graphs) and thus of the claw-free
graphs~\cite{Beineke70}. Consequently, the hardness results we obtain for Kőnig graphs hold for
the three previous classes of graphs as well.

Being a well-studied object for almost a whole century \cite{Whitney92}, line graphs
have applications in both graph theory (see for instance \cite{ChudnovskyRST06}), but also in
algorithm design (see for instance \cite{Marx10}). The first thorough study of {\em homomorphisms between
line graphs} is due to \nesetril~\cite{Nesetril71}; in particular, \nesetril~gave criteria
when a homomorphism from $L(H)$ to $L(G)$ corresponds to a homomorphism from~$H$ to~$G$.
We further motivate the study of line graphs (and by extension Kőnig graphs) by demonstrating
that the problem of {\em finding} a homomorphism to a line graph is always fixed-parameter
tractable:

\begin{theorem}\label{thm:dec_lines_tract_intro}
    The decision problems $\homsprob{\top}{\mathcal{L}}$ and thus $\homsprob{\top}{\symking}$
    are fixed-parameter tractable.
    In particular, given a graph $H$ and a line graph $L$, it is possible to decide the existence
    of a homomorphism from $H$ to $L$ in time \[f(|V(H)|) \cdot O(|V(L)|^2),\] for some computable
    function $f$ independent of $H$ and $L$.\lipicsEnd
\end{theorem}
As it turns out, in contrast, {\em counting} all homomorphisms to Kőnig graph is in general
$\#\W1$-hard; specifically, we prove the following:
\begin{mtheorem}\label{thm:main_line_graphs_intro}
    Let $\mathcal{H}$ denote a recursively enumerable class of graphs.
    If $\mathcal{H}$ has unbounded treewidth and is closed under taking minors, then
    the problem $\#\homsprob{\mathcal{H}}{\symking}$ is $\#\W{1}$-complete.\lipicsEnd
\end{mtheorem}
As argued before, the choice of Kőnig graphs induces the following consequences for perfect and
claw-free graphs.
\begin{corollary}\label{cor:minor_closed_classification_intro}
    Let $\mathcal{C}$ denote one of the classes of line-graphs, claw-free graphs or perfect graphs,
    or a non-empty union thereof.
    Further, let $\mathcal{H}$ denote a recursively enumerable class of graphs.
    \begin{enumerate}[~1~]
        \item If the treewidth of the class $\mathcal{H}$ is bounded, then the problem
            $\#\homsprob{\mathcal{H}}{\mathcal{C}}$ is solvable in polynomial time.
        \item Otherwise, if the class $\mathcal{H}$ is additionally minor-closed, the problem
            $\#\homsprob{\mathcal{H}}{\mathcal{C}}$ is $\#\W{1}$-complete.\lipicsEnd
    \end{enumerate}

\end{corollary}
Note that by restricting the graph class $\mathcal{H}$ in \cref{thm:main_line_graphs_intro},
we are able to give a more detailed view of what makes counting homomorphisms to Kőnig
graphs {\em hard}. In particular, the explicit criterion established by
\cref{thm:main_line_graphs_intro} suggests that we may hope for fast
algorithms only for classes $\mathcal{H}$ that are not minor-closed.

However, in general, \cref{thm:main_line_graphs_intro} does not answer the question
whether we can classify the problem $\#\homsprob{\mathcal{H}}{\symking}$ into easy
and hard problems, hence we answer that question with the following {\em implicit}
dichotomy:

\begin{mtheorem}\label{thm:implicit_line_dicho_intro}
    Let $\mathcal{H}$ denote a recursively enumerable class of graphs.
    Then the problem $\#\homsprob{\mathcal{H}}{\symking}$ is either fixed-parameter tractable
    or $\#\W{1}$-complete under parameterized Turing-reductions.\lipicsEnd
\end{mtheorem}

\paragraph*{Technical Overview}

For our universality result (\cref{ladner_intro}), we rely on known results for
homomorphisms between {\em Kneser graphs}.
More precisely, we use a computable function that associates each integer $n\geq 3$ with a
Kneser graph $\K(n)$ such that there are no homomorphisms between $\K(n)$ and $\K(m)$ whenever
$n\neq m$.
Now, given some problem $P \in \#\W{1}$, we use the existence of a
certain parameterized weakly parsimonious reduction $\mathbb{A}$ from $P$ to the problem of
counting homomorphisms from Kneser graphs $\K(n)$.
In particular, for any instance $x$ of the problem $P$, we have an efficiently computable mapping
to a pair of graphs $x \mapsto (H_x,G_x)$ such that $P(x)$ is equal (up to a normalizing factor)
to the number of homomorphisms from the Kneser graph $H_x$ to the graph $G_x$,
the latter of which can be assumed to allow a homomorphism to $H_x$.
The main idea is then to choose $\mathcal{H}$ as the set of all graphs $H_x$ and $\mathcal{G}$
as the set of all graphs $G_x$. Then we prove that $P$ and
$\#\homsprob{\mathcal{H}}{\mathcal{G}}$ are interreducible.
While the reduction $P \fptred \#\homsprob{\mathcal{H}}{\mathcal{G}}$ is immediate,
we consider the construction of the backward reduction as our main technical contribution.

In particular, consider a pair of graphs $(H_x, G_y) \in \mathcal{H}\times\mathcal{G}$.
In order to obtain a reduction $\#\homsprob{\mathcal{H}}{\mathcal{G}} \fptred P$,
that is to compute the number $\#\homs{H_x}{G_y}$, we need to construct an instance to
the problem $P$. This is easy if both $H_x$ and $G_y$ indeed correspond to the same
instance $z = x = y$. If $H_x$ and $G_y$ do not correspond to a common instance, however,
we need information about $\#\homs{H_x}{G_y}$ from {\em somewhere else}, as any oracle to the
problem $P$ is useless in this situation. In our case,
the construction ensures that $\#\homs{H_x}{G_y} = 0$ in this situation; but in order
to obtain this equality (while also maintaining decodability of the original instance $x$),
an involved construction using Kneser graphs seems to be required.

An even more fundamental (but easier to solve) challenge is to {\em reversibly} encode
any string $x$ into a part of the graph $G_x$ in such a way, that the number of homomorphisms
to the graph $G_x$ changes in a controlled way (in our case the number of homomorphisms stays
in fact the same). As our constructed Kneser graphs have a chromatic number of at least 3,
encoding a string $x$ is possible using a comparably simple construction using paths.
Implicitly, this step as well relies on deep theory about Kneser graphs, in particular
we rely on \lovasz' seminal result~\cite{Lovasz78} which asserts that $H_x$ cannot be mapped
homomorphically into a graph with low chromatic number.\vspace*{1ex}

For our dichotomy results, as advertised, from a technical point of view,
obtaining \cref{thm:genhomdich_intro} is a rather simple lifting exercise from the result
in~\cite{DalmauJ04}; we obtain \cref{thm:subs_f_colored_intro} by a rather
direct application of complexity monotonicity~\cite{CurticapeanDM17}.

In contrast, the analysis of the complexity of counting homomorphisms to Kőnig graphs
is technically more involved.
The proof of the explicit classification for minor-closed classes~$\mathcal{H}$
(\cref{thm:main_line_graphs_intro}) uses a gadget construction that, intuitively, associates
each graph with a Kőnig graph, while keeping the number of grid-like substructures stable.
In view of the diverse applications of the Grid-Tiling Problem
(see for instance\ \cite[Chapter~14.4.1]{CyganFKLMPPS15}), the construction might yield further
intractability results for counting problems on Kőnig graphs (and thus on claw-free and perfect
graphs) and might hence be of independent interest.

Finally, the implicit and exhaustive classification for counting homomorphisms to Kőnig graphs
(\cref{thm:implicit_line_dicho_intro}) relies on Whitney's Isomorphism Theorem for line
graphs~\cite{Whitney92} which allows to express the number of homomorphisms from a graph $H$
to a Kőnig graph $G$ as a finite linear combination of homomorphisms of the form
\[ \sum_{F} \lambda_F \cdot \#\homs{F}{L^{-1}(G)},\]
where the graphs $F$ can be efficiently computed from $H$
and $\#\homs{F}{L^{-1}(G)}$ is the number of
homomorphisms from $F$ to the primal graph of $G$.
\Cref{thm:implicit_line_dicho_intro} then follows by complexity
monotonicity~\cite{CurticapeanDM17} and the classification of counting homomorphisms to
bipartite graph as given by \cref{thm:genhomdich_intro}.

\clearpage
\paragraph*{Organization of the Paper}

We start with an introduction to the concepts and notation used in this work
(including formal definitions of (parameterized) promise problems and reductions between them)
in \cref{sec:prelims}.

The proof of \cref{ladner_intro} is presented in \cref{sec:ladnersection}. For completeness,
we provide a sketch of the hardness proof of $\#\homsprob{\mathcal{H}}{\top}$ in
\cref{sec:hom_hard_strategy}. Continuing, in \cref{sec:fcol} we prove the dichotomy for counting
homomorphisms and subgraphs in $F$-colorable graphs, some proofs are deferred to the
appendix \cref{sec:app_fcol}.
Finally, in \cref{sec:line_graphs}, we present the new dichotomy for Kőnig graphs.

\section{Preliminaries}\label{sec:prelims}
We write $[n]$ to denote the set $\{1,\dots,n\}$. Further, we assume the binary alphabet $\{0,1\}$.
In particular, we assume that numbers are encoded binary as well,
which allows us to abuse notation and write $\mathbb{N} \subseteq \{0,1\}\ast$.
Given a function $g: A \times B \rightarrow C$ and an element $a \in A$, we write $g(a,\star)$ for
the function which maps $b \in B$ to $g(a,b)$.
Given a finite set $A$ we write $|A|$ and $\#A$ for the cardinality of $A$.
Given two functions $f: A \rightarrow B$ and $g:B \rightarrow C$, we write $f\circ g$ for their
\emph{composition} that maps $x \in A$ to $g(f(x)) \in C$.
A \emph{partition} of a set $A$ is a set of non-empty and pairwise disjoint subsets of $A$,
called \emph{blocks}, whose union is $A$.

\subsection{Graphs and Homomorphisms}
We consider undirected simple graphs without self-loops (unless stated otherwise)
and we assume that graphs are encoded by their adjacency matrices.
Given a graph $G$, we write $V(G)$ and $E(G)$ for the vertices and edges of $G$.
A graph is called \emph{complete} or a \emph{clique} if all vertices are pairwise adjacent.
A \emph{subgraph} of $G$ is a graph obtained from $G$ by deleting vertices (including adjacent
edges) and edges; more precisely, the graph $F$ is a subgraph of $G$ if $F=(V,E)$ such that
$V\subseteq V(G)$ and $E \subseteq E(G)\cap V^2$.
The graph $F$ is called a \emph{proper subgraph} if $F \neq G$.
A graph $M$ is a \emph{minor} of a graph $G$ if $M$ can be obtained by a sequence of
edge-contraction from a subgraph of $G$.
Here the contraction of an edge $e=\{u,v\}$ is the operation of adding a new vertex $uv$ which
is made adjacent to all vertices that have been adjacent to $u$ or $v$.
After that, the vertices $u$, $v$ and possible self-loops and multi-edges are deleted.
Given a subset $S$ of vertices of $G$, the \emph{induced subgraph} $G[S]$ has vertices $S$
and edges $E(G)\cap S^2$.
Given a partition $\rho$ of $V(G)$, the \emph{quotient graph} $G/\rho$ of $G$ is obtained
from $G$ by identifying every pair of vertices that are contained in the same block of $\rho$.
After that, multiple edges are deleted. Note that this construction induces self-loops if there
is an edge between two vertices in the same block. Adopting the notation
of~\cite{CurticapeanDM17}, we denote quotient graphs without self-loops as \emph{spasms}.

Given graphs $H$ and $G$, a \emph{homomorphism} from $H$ to $G$ is a mapping
$h:V(H)\rightarrow V(G)$ that preserves the adjacency of vertices, that is,
for every edge $\{u,v\}$ in $E(H)$, the graph $G$ has the edge $\{h(u),h(v)\} \in E(G)$.
If a homomorphism $h$ is injective, then it is called an \emph{embedding}.
If an embedding $h$ additionally satisfies that for every edge $\{h(u),h(v)\}$
in $E(G)$ there is an edge $\{u,v\}$ in $E(H)$, then $h$ is called a \emph{strong embedding}
and a strong embedding is called an \emph{isomorphism} if it is bijective.
Two graphs $H$ and $G$ are called \emph{isomorphic} if there is an isomorphism from $H$ to $G$.
In this paper, we (implicitly) work only on isomorphism classes of graphs; we abuse notation and
write $H = G$ if $H$ and $G$ are isomorphic.
In particular, we denote $\top$ for the set of all (isomorphism types of) graphs.
A homomorphism from $H$ to itself is called an \emph{endomorphism}. Further, a bijective
endomorphism is called an \emph{automorphism}.
We write $\homs{H}{G}$, $\embs{H}{G}$ and $\strembs{H}{G}$ for the set of all homomorphisms,
embeddings and strong embeddings from $H$ to $G$, respectively.
Further, we write $\auts{H}$ for the set of automorphisms of $H$, $\subs{H}{G}$ for the set
of subgraphs of~$G$ that are isomorphic to $H$ and $\indsubs{H}{G}$ for the set of induced subgraphs of $G$ that are isomorphic to $H$.

\paragraph*{Homomorphic Equivalence and Cores}

Two graphs $H$ and $G$ are called \emph{homomorphically equivalent} if there is both a
homomorphism from $H$ to $G$ and a homomorphism from $G$ to $H$.
Clearly, homomorphic equivalence is an equivalence relation; further the following is known.

\begin{lemma}[See for instance\ Chapter~1.6 in~\cite{HellN04}]
    The minimal representative of an equivalence class (with respect to homomorphic equivalence)
    is unique up to isomorphisms.\lipicsEnd
\end{lemma}
It is thus well-defined to speak of ``the'' minimal representative of an equivalence class;
we say that such a graph is a \emph{core}.
An equivalent definition of a core is given by the following known result.
\begin{lemma}[See for instance\ Chapter~1.6 in~\cite{HellN04}]\label{lem:coraut}
    A graph $H$ is a core if and only if it there is no homomorphism from $H$ to a proper
    subgraph of $H$. In particular, every endomorphism of a core is an
    automorphism.\lipicsEnd
\end{lemma}

\paragraph*{Colorings and Graph Parameters}

An $H$\emph{-coloring} of a graph $G$ is a homomorphism $c \in \homs{G}{H}$.
We say that a graph $G$ is $H$\emph{-colorable} or allows a coloring into $H$
if the graph $G$ has an $H$-coloring.
In particular, given a positive integer $k \in \mathbb{N}$, we say that $G$ is $k$-colorable if it
allows a coloring into the complete graph with $k$ vertices.
Given a graph $G$ together with an $H$-coloring $c$,
we say that a homomorphism $h \in \homs{H}{G}$ is \emph{color-prescribed} if $c(h(v))= v$ for all $v \in V(H)$.
We write $\#\cphoms{H}{G}$ for the set of all color-prescribed homomorphisms from $H$ to $G$.

The following three graph parameters are of particular importance in this paper.
First, the \emph{chromatic number} of a graph $G$ is defined to be the smallest $k$ such that $G$ is $k$-colorable.
Second, the \emph{odd girth} is defined to be the length of the smallest odd cycle in a graph, and undefined if no odd cycle exists.
Third, we rely on the graph parameter of \emph{treewidth}.
Intuitively, a graph $G$ has small treewidth if it has a ``tree-like structure''.
In particular, if a graph $G$ has a small treewidth, then the graph $G$ also has small
``separators''. These separators allow for efficient dynamic programming algorithms for a wide range of problems that are known to be hard in the unrestricted setting.
However, as we need the treewidth of a graph only in a black-box manner, we defer the reader to
the literature (for instance\ Chapter~7 in~\cite{CyganFKLMPPS15}) for the formal definition and a detailed
exposition.

\paragraph*{Tensor Products of Graphs}
Given two graphs $G$ and $A$, the \emph{tensor product} $G\times A$ is the graph with vertices
$V(G)\times V(A)$, where $\times$ is the Cartesian product of sets. Two vertices $(g,a)$ and
$(g',a')$ are adjacent in $G\times A$ if the edge $\{g,g'\}$ is in $E(G)$ and the edge $\{a,a'\}$
is in $E(A)$. Now let $H$ be a fixed graph.
It is well-known that the function $\#\homs{H}{\star}$ is linear with respect to $\times$
and multiplication~\cite[Equation~5.30]{Lovasz12}, that is,
\[\#\homs{H}{G \times A} = \#\homs{H}{G} \cdot \#\homs{H}{A}.\]
A further well-known fact about the tensor product reads as follows.
\begin{fact}[Folklore]\label{fac:f_col_tensor}
    Let $G$, $A$ and $F$ be graphs. If either one of $G$ or $A$ is $F$-colorable,
    then so is their tensor product $G \times A$.\lipicsEnd
\end{fact}

\subsection{Parameterized Counting Problems and Promises}
Recall that the problem $\#\homsprob{\mathcal{H}}{\mathcal{G}}$ asks, given graphs
$H \in \mathcal{H}$ and $G \in \mathcal{G}$, to compute the number of homomorphisms from
$H$ to $G$. However, this definition is informal in the sense that it leaves out the specification
of the output if the input is {\em invalid}, that is, if the graph $H$ does not belong to
the class $\mathcal{H}$ or the graph $G$ does not belong to the class $\mathcal{G}$.

A naive option to solve this issue, is to require the output to be $0$ if the input is invalid.
However, in this case we exclude a plethora of interesting cases from our studies,
as even seemingly trivial instances might encode $\cc{NP}$-hard problems.
Consider for example the class $\mathcal{G}$ of $3$-colorable graphs.
Following the naive option, the problem $\#\homsprob{\mathcal{H}}{\mathcal{G}}$ becomes
$\cc{NP}$-hard even if the class $\mathcal{H}$ contains only the graph $K_1$ consisting of a
single vertex, as it encodes the $3$-colorability problem:
An instance $(K_1,G)$ of the problem $\#\homsprob{\mathcal{H}}{\mathcal{G}}$ is mapped to zero
if and only if the graph $G$ is $3$-colorable.
In particular, fixed-parameter tractability of this problem would yield an algorithm running in
time $f(|K_1|)\cdot |V(G)|^{O(1)}$ which is a polynomial in $|V(G)|$, and thus imply
$\cc{P}=\cc{NP}$.
In sharp contrast, the number of homomorphisms from the graph $K_1$ to any graph $G$ is just
the number of vertices $|V(G)|$ and thus the hardness of the problem
$\#\homsprob{\mathcal{H}}{\mathcal{G}}$ stems only from enforcing invalid inputs to be mapped
to zero.

Another option for solving the issue of invalid instances of the problem
$\#\homsprob{\mathcal{H}}{\mathcal{G}}$ is as follows: If a given instance $(H, G)$
consists of graphs $H \in \mathcal{H}$ and $G \in \mathcal{G}$, then we are supposed
to compute the number of homomorphisms $\#\homs{H}{G}$ correctly; otherwise,
we may output {\em any} number.
Formally, this requires us to model the problem $\#\homsprob{\mathcal{H}}{\mathcal{G}}$
as a {\em promise problem}.\footnote{Goldreich~\cite[Chapter~2.4.1]{Goldreich08} states that
\emph{promise problems offer the most direct way of formulating natural computational problems}.
Indeed, some of the most striking results in complexity theory implicitly rely on promise
problems. Examples are ``gap problems'' and ``uniqueness promises''; we refer the reader to
\cite[Chapter~2.4.1.2]{Goldreich08} for a discussion.}
In what follows, we thus present a concise but self-contained introduction to parameterized
promise (counting) problems.

A \emph{parameterization} $\kappa$ is a polynomial-time computable function from the set $\{0,1\}\ast$
to the natural numbers $\mathbb{N}$.
Note that the assumption of polynomial-time computable parameterizations for both, decision and
counting problems, is common (see for instance\ \cite[Definitions~1.1 and~14.1]{FlumG06}), but not
standard. We refer the reader to the discussion of this issue in Chapter~1.2 in the textbook of Flum and Grohe~\cite{FlumG06}.

\begin{definition}\label{def:ppc}
    A \emph{parameterized promise} \emph{counting problem} (or short ``PPC problem'') is a
    triple~$(P,\kappa,\Pi)$ consisting of a function $P:\{0,1\}\ast \rightarrow \mathbb{N}$,
    a parameterization $\kappa: \{0,1\}\ast \rightarrow \mathbb{N}$, and a
    promise~$\Pi \subseteq \{0,1\}\ast$.
    A \emph{parameterized counting problem} (without promises) is a PPC problem with
    $\Pi = \{0,1\}\ast$.

    A PPC problem $(P,\kappa,\Pi)$ is computable in time $t$ if there is a deterministic
    algorithm~$\mathbb{A}$ that fulfills the following. %for which the following holds true:
    \begin{enumerate}
        \item On input $x \in \{0,1\}\ast$, the algorithm $\mathbb{A}$ runs in time $t(|x|)$.
        \item On input $x \in \Pi$, the algorithm $\mathbb{A}$ outputs $P(x)$.
    \end{enumerate}
    In particular, we call the problem $(P,\kappa,\Pi)$ \emph{fixed-parameter tractable} if there
    is a computable function~$f$ such that the triple $(P,\kappa,\Pi)$ can be computed in
    time $f(\kappa(x))\cdot |x|^{O(1)}$.

    Further, we call $x \in \{0,1\}\ast$ an {\em instance} of the problem  $(P, \kappa, \Pi)$ and
    say that an instance $x$ is {\em valid} if it is contained in the promise, that is $x
    \in \Pi$.\lipicsEnd
\end{definition}
Note that we obtain the standard definition of fixed-parameter tractability of (non-promise)
counting problems if we set the promise $\Pi$ to be $\{0,1\}\ast$.
Note further that parameterized decision problems with promises are obtained from \cref{def:ppc}
by restricting the image of the function $P$ to be $\{0,1\}$.
In this case, \cref{def:ppc} coincides with the standard definition of (parameterized) promise
problems (see for instance\ Definition~3.1 in the full version~\cite{BGS18arxiv} of~\cite{BGS18}).

We consider the following family of PPC problems.
\begin{definition}
    Let $\mathcal{H}$ and $\mathcal{G}$ be classes of graphs.
    The PPC problem $\#\homsprob{\mathcal{H}}{\mathcal{G}}$ asks,
    given $H \in \mathcal{H}$ and $G \in \mathcal{G}$, to compute the number of homomorphisms
    $\#\homs{H}{G}$; the parameter is $|V(H)|$.
    Formally, the promise of $\#\homsprob{\mathcal{H}}{\mathcal{G}}$ is the set of all (encodings of) pairs $(H,G)\in \mathcal{H}\times \mathcal{G}$.

    \noindent Further, we define $\#\cphomsprob{\mathcal{H}}{\mathcal{G}}$ as the PPC problem of,
    given $H \in \mathcal{H}$, $G \in \mathcal{G}$ and an $H$-coloring of $G$, computing
    the number of color-prescribed homomorphisms $\#\cphoms{\mathcal{H}}{\mathcal{G}}$;
    the parameter is $|V(H)|$.
    Again, the formal promise is defined as the set of all (encodings of) pairs $(H,G)\in
    \mathcal{H}\times \mathcal{G}$.\lipicsEnd
\end{definition}
The \emph{decision problems} $\homsprob{\mathcal{H}}{\mathcal{G}}$ and
$\cphomsprob{\mathcal{H}}{\mathcal{G}}$ are defined similarly, with the exception that the output is required to be $1$ if the number of homomorphisms $\#\homs{H}{G}$ is positive or
the number of color-prescribed homomorphisms $\#\cphoms{H}{G}$ is positive, respectively.

\begin{remark}\label{rem:nopromise}
    If membership of a graph in the class $\mathcal{G}$ can be tested in polynomial time and
    the class $\mathcal{H}$ is recursive, then there is no need to define the problem
    $\#\homsprob{\mathcal{H}}{\mathcal{G}}$ as promise problem.
    Instead, we can define the output to be zero if a given pair $(H,G)$ is not contained in
    $\mathcal{H}\times \mathcal{G}$; note that $H \in \mathcal{H}$ can be verified in time $f(H)$
    for some computable function $f$ as, by assumption, $\mathcal{H}$ is
    recursive.\lipicsEnd
\end{remark}

\paragraph*{Reductions and Hardness}

In this paper, we consider the following two notions of reducibility for PPC problems.
\begin{definition}[Parameterized (Weakly) Parsimonious Reductions]
    Let PPC problems $(P,\kappa,\Pi)$ and $(P',\kappa',\Pi')$ be given.
    A \emph{parameterized weakly parsimonious reduction} from $(P,\kappa,\Pi)$ to
    $(P',\kappa',\Pi')$ is a pair consisting of a deterministic algorithm $\mathbb{A}$ and a
    triple of computable functions $(f,g,s)$ such that:
    \begin{enumerate}[~1~]
        \item For all valid instances $x \in \Pi$,
            the algorithm $\mathbb{A}$ outputs a valid instance of
            $(P',\kappa',\Pi')$, that is $\mathbb{A}(x)\in \Pi'$.
        \item Given a valid instance $x \in \Pi$, we can compute the value $P(x)$
            as the product of the value $g(x)$ and the solution of
            the computed instance $\mathbb{A}(x)$ to problem $P'$;
            that is, we have that $P(x) = g(x) \cdot P'(\mathbb{A}(x))$.
        \item The PPC problem $(g,\kappa,\Pi)$ is fixed-parameter tractable.
        \item On input $x \in \{0,1\}\ast$, the algorithm $\mathbb{A}$ runs in time
            $f(\kappa(x))\cdot |x|^{O(1)}$.
        \item For all $x \in \{0,1\}\ast$, the parameter of the instance $\mathbb{A}(x)$
            is bounded by $s(\kappa(x))$,  that is $\kappa'(\mathbb{A}(x)) \leq s(\kappa(x))$.
    \end{enumerate}
    We write $(P,\kappa,\Pi)\wparsired (P',\kappa',\Pi')$ if such a reduction exists.
    If $g$ is the identity function on $\Pi$, then the reduction is called \emph{parsimonious}
    and we write $(P,\kappa,\Pi)\parsired (P',\kappa',\Pi')$.\lipicsEnd
\end{definition}

\begin{definition}[Parameterized Turing-reductions]
    Let $(P,\kappa,\Pi)$ and $(P',\kappa',\Pi')$ be PPC problems.
    A \emph{parameterized Turing-reduction} from $(P,\kappa,\Pi)$ to $(P',\kappa',\Pi')$
    is a pair of an algorithm $\mathbb{A}$ equipped with oracle access to the function $P'$
    and a pair $(f,s)$ of computable functions such that:
    \begin{enumerate}[~1~]
        \item On input $x \in \{0,1\}\ast$,
            the algorithm $\mathbb{A}$ runs in time $f(\kappa(x)) \cdot |x|^{O(1)}$.
        \item On input $x \in \Pi$,
            the algorithm $\mathbb{A}$ computes the function  $P(x)$.
        \item On input $x \in \Pi$,
            the algorithm $\mathbb{A}$ queries the oracle only on strings $y$ with $y\in \Pi'$
            and $\kappa'(y)\leq s(\kappa(x))$.
    \end{enumerate}
    We write $(P,\kappa,\Pi)\fptred (P',\kappa',\Pi')$ if such a reduction
    exists.\lipicsEnd
\end{definition}
Unsurprisingly, the previous notions of reducibility coincide with the common notions for
reducibility between parameterized counting problems if the promises $\Pi$ and $\Pi'$ are trivial,
that is, if $\Pi= \Pi'=\{0,1\}\ast$ (see for instance\ \cite[Definition~1.8]{Curticapean15}).
Further, the following facts are straightforward to verify.
\begin{fact}
    Let $(P,\kappa,\Pi)$ and $(P',\kappa',\Pi')$ be PPC problems. We have that
    \[(P,\kappa,\Pi)\parsired (P',\kappa',\Pi') \Longrightarrow
    (P,\kappa,\Pi)\wparsired (P',\kappa',\Pi') \Longrightarrow
(P,\kappa,\Pi)\fptred (P',\kappa',\Pi'). \]

Further, all of the notions of reducibility $\parsired$, $\wparsired$, and $\fptred$ are
transitive.\lipicsEnd
\end{fact}

\begin{fact}
Let $(P,\kappa,\Pi)$ and $(P',\kappa',\Pi')$ be PPC problems and assume that $(P,\kappa,\Pi)$
reduces to $(P',\kappa',\Pi')$ with respect to any of $\parsired$, $\wparsired$, or $\fptred$.
If $(P',\kappa',\Pi')$ is fixed-parameter tractable, then  $(P,\kappa,\Pi)$ is also
fixed-parameter tractable.\lipicsEnd
\end{fact}

Evidence of fixed-parameter \emph{in}tractability of parameterized counting problems
(with promises) is given by hardness for the complexity class $\#\W{1}$.
It is common to define $\#\W{1}$ via the complete problem $\#\clique$.
The problem $\#\clique$ asks, given $k \in \mathbb{N}$ and a graph $G$, to compute the number of
cliques of size $k$ in $G$.
Note that by \cref{rem:nopromise} the problem $\#\clique$ can be assumed to have no promise.

\begin{definition}[\cite{FlumG04,McCartin06}]
    The class $\#\W{1}$ contains all parameterized counting problems without promises that can be
    reduced to $\#\clique$ by parameterized parsimonious reductions.\lipicsEnd
\end{definition}
Next, we extend the notion of $\#\W{1}$-hardness to PPC problems.
\begin{definition}
    We say a PPC problem $(P,\kappa,\Pi)$ is $\#\W{1}$\emph{-hard} under parameterized parsimonious
    reductions if
    \[ \#\clique \parsired (P,\kappa,\Pi).\]
    Hardness under $\wparsired$ and $\fptred$ is defined likewise.
    A parameterized counting problem without promises is $\#\W{1}$\emph{-complete} if it is
    contained in $\#\W{1}$ and $\#\W{1}$-hard.\lipicsEnd
\end{definition}

As $\#\homsprob{\mathcal{H}}{\mathcal{G}}$ is a promise problem, it is formally not contained in
$\#\W{1}$. However, it can be shown that $\#\homsprob{\mathcal{H}}{\mathcal{G}}$ cannot be harder
than $\#\W{1}$-complete problems:
\begin{lemma}
    Let $\mathcal{H}$ and $\mathcal{G}$ be computable graph classes. We have that
    \[\#\homsprob{\mathcal{H}}{\mathcal{G}} \parsired \#\clique.\]
\end{lemma}
\begin{proof}
    It is known that $\#\homsprob{\mathcal{H}}{\top} \parsired \#\clique$ by the more general
    result that $\#\A{1} = \#\W{1}$~\cite[Theorem~14.17]{FlumG06}.
    The reduction $\#\homsprob{\mathcal{H}}{\mathcal{G}} \parsired \#\homsprob{\mathcal{H}}{\top}$
    is given by the identity function.
\end{proof}
As a concluding remark for this subsection, note that parameterized reductions, $\#\clique$, and
(hardness and completeness for) $\#\W{1}$ have corresponding notions in the decision realm.
In particular, $\clique$, that is, the problem of \emph{deciding} the existence of a clique of
size~$k$, constitutes the canonical complete problem for $\W{1}$.
We refer the reader to the textbook of Flum and Grohe~\cite{FlumG06} for further details of parameterized decision complexity, as this work mainly deals with counting problems.

\subsection{Quantum Graphs and Complexity Monotonicity}\label{sec:quantum}

The framework of {\em Complexity Monotonicity} was recently introduced by
Curticapean, Dell, and Marx in their breakthrough result regarding the complexity of the (induced)
subgraph counting problem~\cite{CurticapeanDM17}.
Very roughly speaking, the principle of complexity monotonicity states that
\begin{center}
    \textit{Computing a linear combination of homomorphism numbers is precisely as hard as computing its hardest term.}
\end{center}
While linear combinations of homomorphisms have been modeled by so-called
\emph{graph motif parameters} in~\cite{CurticapeanDM17}, we instead rely on the notion of
\emph{quantum graphs} as introduced by \lovasz~\cite[Chapter~6]{Lovasz12}.
\begin{definition}[Quantum graphs]
    A \emph{quantum graph} $Q$ is a formal linear combination of graphs with finite support.
    We write \[Q = \sum_H \lambda_H \cdot H,\]
    where $\lambda_H$ is non-zero only for finitely many graphs.
    We write $\mathsf{supp}(Q)$ for the set of all graphs $H$ for which $\lambda_H$ is non-zero.
    The elements of the support $\mathsf{supp}(Q)$ are called \emph{constituents}
    of~$Q$.\lipicsEnd
\end{definition}
Graph parameters extend to quantum graphs linearly. In particular, we define
\begin{equation}\label{eq:qgraphs}
    \#\homs{Q}{G} := \sum_H \lambda_H \cdot \#\homs{H}{G}.
\end{equation}
Now, given a set $\mathcal{Q}$ of quantum graphs, we write $\supp(\mathcal{Q})$ for the set of all
constituents of all quantum graphs in $\mathcal{Q}$.
Further, given a class $\mathcal{G}$ of graphs, the PPC problem
$\#\homsprob{\mathcal{Q}}{\mathcal{G}}$ is defined similarly as in case of (non-quantum) graphs:
Given $(Q,G)\in \mathcal{Q}\times \mathcal{G}$, the goal is to compute the number $\#\homs{Q}{G}$;
the parameter is given by the description length $|Q|$ of $Q$.

The main result of Curticapean, Dell and Marx can be stated as follows:
\begin{theorem}[Complexity Monotonicity~\cite{CurticapeanDM17}]
    Let $\mathcal{Q}$ be a recursively enumerable class of quantum graphs. Then we have that
    \[\#\homsprob{\mathcal{Q}}{\top} \interred
    \#\homsprob{\mathsf{supp}(\mathcal{Q})}{\top}.\]\lipicsEnd
\end{theorem}
The reduction
$\#\homsprob{\mathcal{Q}}{\top} \fptred \#\homsprob{\mathsf{supp}(\mathcal{Q})}{\top}$ is trivial:
Given a quantum graph $Q$ and a graph $G$, we can compute the number $\#\homs{Q}{G}$ as given by
Equation~\eqref{eq:qgraphs}.
However, the other direction relies on a deep theory of
\lovasz~\cite[Chapters~5 and~6]{Lovasz12} and is given by the following lemma.
\begin{lemma}[Lemma~3.6 in~\cite{CurticapeanDM17}]\label{lem:monotonicity}
    Let $Q$ be a quantum graph.
    There is a deterministic algorithm $\mathbb{A}$ that is given oracle access to
    $\#\homs{Q}{\star}$ and that, on input a graph $G$, computes the number $\#\homs{H}{G}$ for
    every constituent $H$ of $Q$. Further, there are computable functions $f$ and $s$ such
    that the running time of $\mathbb{A}$ is bounded by $f(|Q|)\cdot |V(G)|^{O(1)}$ and the number
    of vertices of every graph $G'$ for which the oracle is queried,
    is bounded by $s(|Q|)\cdot |V(G)|$.\lipicsEnd
\end{lemma}
In \cref{sec:fcol} we show that the previous lemma readily extends to the problem
$\#\homsprob{\mathcal{Q}}{\mathcal{G}_F}$, where $\mathcal{G}_F$ is the set of all $F$-colorable
graphs for some fixed graph $F$.

\section{Counting and Finding Homomorphisms is Universal}\label{sec:ladnersection}
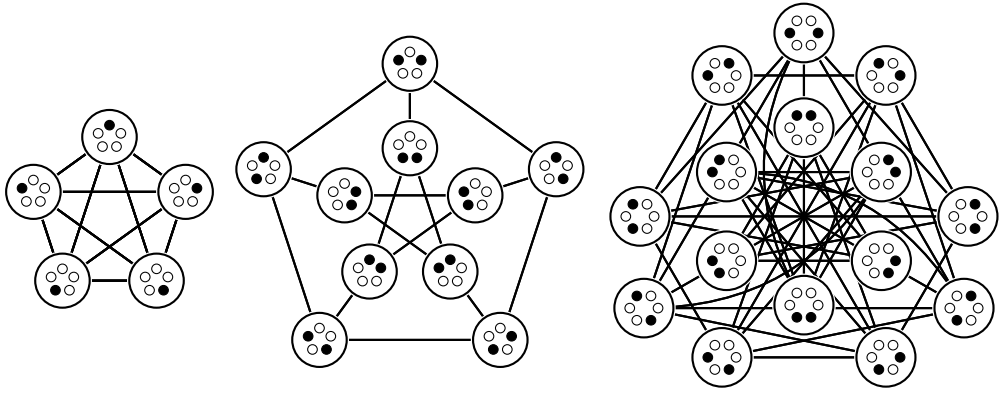
\begin{figure}
    \centering
    \begin{tikzpicture}[scale=.85,transform shape]
        \pic at (0,0) {kneser=5/1};
        \pic at (4.65,0) {kneser=5/2};
        \pic[edge fix=2] at (10.75,0) {kneser=6/2};
    \end{tikzpicture}
    \caption{Examples for Kneser graphs: The Kneser graphs $\K(5,1)$, $\K(5,2)$, and $\K(6,2)$;
    the graph~$\K(5,2)$ is also known as the \emph{Petersen graph}.
    The subset corresponding to each vertex is indicated by white and black dots.}\label{fig:kne}
\end{figure}

In this part of the paper, we show that every parameterized counting problem in $\#\W{1}$ is
interreducible with a problem $\#\homsprob{\mathcal{H}}{\mathcal{G}}$ with respect to parameterized
Turing-reductions.
Further, the proof shows that the analogous statement holds for (parameterized) decision problems
in $\W{1}$ and a problem $\homsprob{\mathcal{H}}{\mathcal{G}}$.
The starting point is the following lemma; it follows from the standard hardness proof of
$\#\homsprob{\mathcal{H}}{\top}$ for classes~$\mathcal{H}$ of unbounded treewidth.
We provide an exposition of the proof in \cref{sec:hom_hard_strategy}---see \cref{lem:help_hardness}.

\begin{lemma}\label{lem:parsihard}
    Let $\mathcal{H}$ be a computable class of connected cores of unbounded treewidth.
    Then the problem $\#\homsprob{\mathcal{H}}{\top}$ is $\#\W{1}$-hard under
    parameterized weakly parsimonious reductions.
    In particular, the images of the reductions can be assumed to contain only pairs~$(H,G)$
    such that $H\in \mathcal{H}$ and $G$ is connected and $H$-colorable.\lipicsEnd
\end{lemma}
The remainder of this section is devoted to the proof of the following theorem.
\begin{theorem}\label{thm:universal_turing_red}
    Let $(F,\kappa)$ denote a problem in $\#\W{1}$.
    There are classes $\mathcal{H}$ and $\mathcal{G}$ such that
    \[(F,\kappa) \interred \#\homsprob{\mathcal{H}}{\mathcal{G}}.\]
    Further, $\mathcal{H}$ is recursively enumerable and $\mathcal{G}$ is
    recursive.\lipicsEnd
\end{theorem}
In particular, we show that for any problem $(F, \kappa)$ in $\#\W{1}$, we can construct graph
classes~$\mathcal{H}$ and~$\mathcal{G}$ such that we have for any graphs $H \in \mathcal{H}$
and $G \in \mathcal{G}$:
\begin{itemize}
    \item If $\#\homs{H}{G}\neq 0$, the pair $(H,G)$ corresponds to exactly one instance
        $x$ of the problem $(F,\kappa)$, and we can obtain both $x$ and $\kappa(x)$ from
        $H$ and $G$.
    \item If $\#\homs{H}{G}=0$, the pair $(H,G)$ does not correspond to an instance
        of $(F,\kappa)$.
\end{itemize}

\paragraph*{Counting Homomorphisms Between Kneser Graphs}

By \cref{lem:parsihard}, the problem $\#\homsprob{\mathcal{H}}{\top}$ is $\#\W1$-hard
for {\em any} (recursively enumerable) class of connected cores
(of unbounded treewidth). In particular, it is known
that this is the case for the
class of {\em Kneser graphs}. Further, Kneser graphs have other nice properties which we
exploit in the proof of \cref{thm:universal_turing_red}.
Formally, we start with the following definition; consider also \cref{fig:kne} for examples of
Kneser graphs.

\begin{definition}[see for instance\ {\cite[Chapter~3]{HahnT97}}]
    Given integers $r$ and $s$, the \emph{Kneser graph} $\K(r,s)$ is the graph that has as vertices
    the subsets of size $s$ of $[r]$ and edges between two vertices if the corresponding sets are
    disjoint.\lipicsEnd
\end{definition}

Given a number $n \geq 3$, we set $\K(n) := \K((2n+1)(n-2), n(n-2))$.
With this choice of the parameters~$r$ and~$s$, we can use the following results;
recall that the chromatic number of a graph $G$ is the minimum $k$ such that $G$ allows a
homomorphism to the complete graph of size $k$, and the odd girth of a graph is the length of the
smallest cycle of odd length.
\begin{fact}[\cite{Lovasz78} and~Propositions~3.13,~3.14 in~\cite{HahnT97}] \label{fac:chromoddgirth}
    The graph $\K(n)$ has chromatic number~$n$ (and thus a treewidth of at least $n-1$)
    and odd girth $2n+1$.
    Further, the graph $\K(n)$ is a core, that is, the graph $\K(n)$ is minimal with respect
    to homomorphic equivalence.
\end{fact}
Note that by \cref{lem:coraut}, the graph $\K(n)$ being a core implies that every endomorphism
of~$\K(n)$ is already an automorphism.
Hence, the number of homomorphisms from the graph~$\K(n)$ to itself is precisely the number of automorphisms $\#\auts{\K(n)}$.

\begin{fact}[Folklore, see for instance\ \cite{VJ05}]
    The graph $\K(r,s)$ is connected if $r > 2s$.\\
    Hence, $\K(n)$ is connected.\lipicsEnd
\end{fact}
An important property of Kneser graphs is the well-known
fact that they constitute an antichain with respect to the homomorphism order.
We provide a proof for convenience.
\begin{lemma}\label{lem:kneser_magic}
    Let $n$ and $m$ be distinct positive integers. Then we have
    \[\#\homs{\K(m)}{\K(n)} = 0.\]
\end{lemma}
\begin{proof}
    For every pair of graphs $H$ and $G$ with $\#\homs{H}{G}\neq 0$, we have that the odd girth
    of $H$ is bounded from below by the odd girth of $G$~\cite[Exercise~1.10.2]{HellN04}.
    Further, the chromatic number of $H$ is bounded from above by the chromatic number of
    $G$~\cite[Proposition~1.8]{HellN04}. The lemma hence holds by \cref{fac:chromoddgirth}.
\end{proof}
Now, let $\mathcal{K}_{even}$ denote the set of all graphs $\K(n)$ with even $n$
and let $\mathcal{K}_{odd}$ denote the set of all graphs $\K(n)$ with odd $n$.

\paragraph*{Encoding Problems into Graphs Classes}

A central tool for the proof of \cref{thm:universal_turing_red} is an encoding of arbitrary
strings into graphs, which we discuss next. In particular, we use a disjoint union of paths
for the encoding.
To this end, let $P_i$ be the path with $i$ edges.
Given a string $x = x[1]x[2]\cdots x[n] \in \{0,1\}\ast$, we define $\enc(x)$ to be the graph
that is the disjoint union of paths $P_i$ for all $i\leq |x|$ with $x_i = 1$, as well as
of $|x|$ isolated vertices. Consider \cref{fig:strenc} for a visualization.

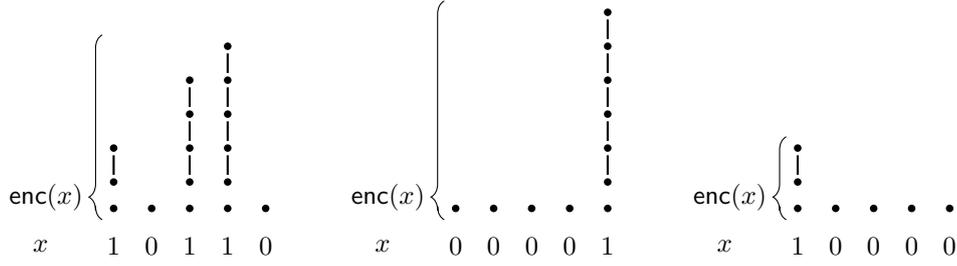
\begin{figure}
    \centering
    \begin{tikzpicture}[scale=.9]
        \pic at (0,0) {encode=5/10110/1/.5};
        \pic at (5,0) {encode=5/1/1/.5};
        \pic at (10,0) {encode=5/10000/1/.5};
    \end{tikzpicture}
    \caption{Examples of strings and their encoding into a graph.}\label{fig:strenc}
\end{figure}

Next, we show how to use the encoding $\enc$ to (reversibly) encode an instance of an
arbitrary problem in $\#\W{1}$ into a pair of graphs. To that end, let $(F,\kappa)$
denote any problem in $\#\W{1}$.
By \cref{lem:parsihard,fac:chromoddgirth}, we have that $(F,\kappa)\parsired\#\homsprob{\mathcal{K}_{even}}{\top}$,
that is, there is an algorithm $\mathbb{A} = \mathbb{A}_F$ and a triple $(f,g,s)$ of computable
functions such that for all strings $x \in \{0,1\}\ast$ all of the following holds:
\begin{enumerate}[(a)]
    \item The algorithm $\mathbb{A}$ computes a pair of graphs $\mathbb{A}(x)=(H_x,G_x)$,
        where $H_x \in \mathcal{K}_{even}$ and the graph $G_x$ is connected and
        $H_x$-colorable.\label{enp1}
    \item The answer to the instance $x$ of the problem $(F,\kappa)$ can be computed as
        \[F(x)=g(x)\cdot \#\homs{H_x}{G_x}.\] \label{enp2}
    \item The problem $(g,\kappa)$ is fixed-parameter tractable, that is
        it is solvable in time
        $f'(\kappa(x))\cdot{|x|}^{O(1)}\!$.\label{enp3}
    \item The algorithm $\mathbb{A}$ runs in time $f(\kappa(x))\cdot |x|^{O(1)}\!$.\label{enp4}
    \item The size of the computed graph $|V(H_x)|$ is at most $s(\kappa(x))$.\label{enp5}
\end{enumerate}

Now, let an instance $x$ to $(F,\kappa)$ be given.
Using the graphs $\mathbb{A}(x) = (H_x,G_x)$ computed by the algorithm $\mathbb{A}$,
we construct a pair of new graphs, which additionally encodes the original instance $x$
as well as its parameter $\kappa(x)$ by setting
\begin{align}
    \hat{H}_x &:= H_x \cup \K(2\kappa(x)+3),\text{ and }\label{eqn:k1}\\
    \hat{G}_x &:= G_x \cup \enc(\langle x, H_x \rangle) \cup \K(2\kappa(x)+3),\label{eqn:k2}
\end{align}
where $\langle x, H_x \rangle$ is any efficient encoding of the pair $(x,H_x)$.
We proceed to show that the constructed graphs $\hat{H}_x$ and $\hat{G}_x$ behave as intended;
also consider \cref{fig:24} for a visualization.

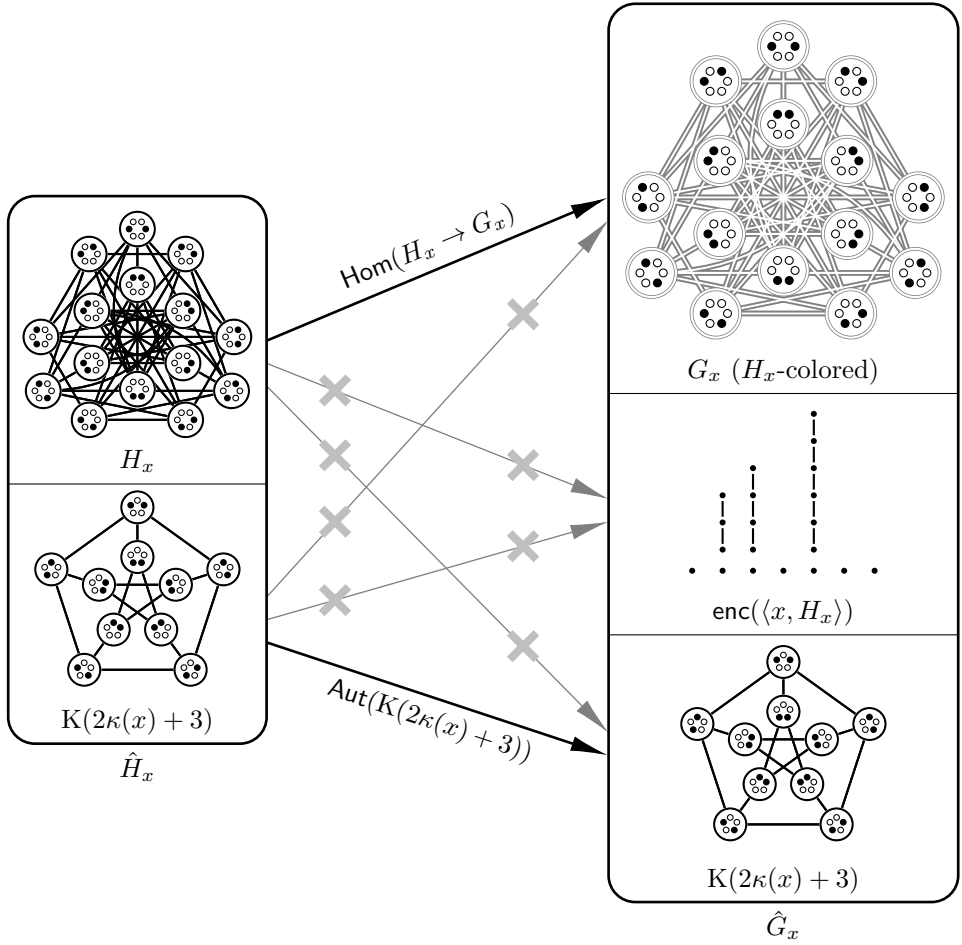
\begin{figure}[t]
    \centering
    \tikzset{
        kns/.pic = {
            \begin{scope}[scale=.5, every node/.append style={transform shape}]
                \pic[edge fix=2] {kneser=6/2};
            \end{scope}
        }
    }
    \tikzset{
        knsb/.pic = {
            \begin{scope}[scale=.7, every node/.append style={transform shape}]
                \pic[white,fill=gray,edge fix=2] {kneser=6/2};
            \end{scope}
        }
    }

    \tikzset{
        knns/.pic = {
            \begin{scope}[scale=.5, every node/.append style={transform shape}]
                \pic {kneser=5/2};
            \end{scope}
        }
    }
    \begin{tikzpicture}

        \draw[line width=1pt,-{Latex[length=5mm, width=2mm]}] (1.7,1.9) --node[above,midway,sloped]
        {$\homs{H_x}{G_x}$} (8.5-2.3,3.8);

        \draw[gray,line width=.5pt,-{Latex[length=5mm, width=2mm]}] (1.7,1.6) --
        node[gray!50,cross=8pt, line width=3pt,pos=.2]{}
        node[gray!50,cross=8pt, line width=3pt,near end]{}
        (8.5-2.3,-.2);
        \draw[gray,line width=.5pt,-{Latex[length=5mm, width=2mm]}] (1.7,1.3) --
        node[gray!50,cross=8pt, line width=3pt,pos=.2]{}
        node[gray!50,cross=8pt, line width=3pt,near end]{}
        (8.5-2.3,-3.3);
        \draw[gray,line width=.5pt,-{Latex[length=5mm, width=2mm]}] (1.7,-1.5) --
        node[gray!50,cross=8pt, line width=3pt,pos=.2]{}
        node[gray!50,cross=8pt, line width=3pt,near end]{}
        (8.5-2.3,3.5);
        \draw[gray,line width=.5pt,-{Latex[length=5mm, width=2mm]}] (1.7,-1.8) --
        node[gray!50,cross=8pt, line width=3pt,pos=.2]{}
        node[gray!50,cross=8pt, line width=3pt,near end]{}
        (8.5-2.3,-.5);

        \draw[line width=1pt,-{Latex[length=5mm, width=2mm]}] (1.7,-2.1) --node[below,
        midway,sloped]{$\auts{\K(2\kappa(x)+3)}$}  (8.5-2.3,-3.6);

        \begin{scope}[local bounding box=a]
            \node at (0, 3.7) {};
            \pic at (0,1.95) {kns};
            \node at (0, .3) {$H_x$};
            \pic at (0,-1.5) {knns};
            \node at (0, -3.15) {$\K(2\kappa(x)+3)$};
            \draw (-1.7,0) -- (1.7,0) {};
            \node at (-1.4,0) {};
            \node at (1.4,0) {};
        \end{scope}
        \draw[rounded corners=10pt,line width=1pt] (a.north east) rectangle (a.south west);
        \node[yshift=-2ex] (a) at (a.south) {$\hat{H}_x$};

        \begin{scope}[local bounding box=b, shift={(8.5,0)}]
            \node at (0,6.25) {};
            \pic at (0,3.8) {knsb};
            \node at (0, 1.5) {$G_x$ ($H_x$-colored)};
            \pic[scale=.8,transform shape] at (0-1.6,-1.15) {encode=7/110100/0/.5};
            \node at (0, -1.7) {$\enc(\langle x, H_x\rangle)$};
            \pic at (0,-3.55) {knns};
            \node at (0, -5.25) {$\K(2\kappa(x)+3)$};

            \draw (-2.3,1.2) -- (2.3,1.2);
            \draw (-2.3,-2) -- (2.3,-2);
            \node at (-1.9,0) {};
            \node at (1.9,0) {};
        \end{scope}

        \draw[rounded corners=10pt,line width=1pt] (b.north east) rectangle (b.south west);
        \node[yshift=-2ex] (a) at (b.south) {$\hat{G}_x$};

    \end{tikzpicture}
    \caption{\Cref{lem:mainp1} illustrated. A cross denotes that no homomorphisms between
        the parts of the graphs exist. Note that the Kneser graphs used in the lemma differ
    from the ones depicted.}\label{fig:24}
\end{figure}

\begin{lemma}\label{lem:mainp1}
    Given a problem $(F,\kappa)$ and an instance $x$ to $(F,\kappa)$, let
    the graphs $\hat{H}_x$ and $\hat{G}_x$ be defined as in \cref{eqn:k1,eqn:k2}.
    Then we have that
    \begin{enumerate}
        \item The graph $\hat{G}_x$ is $\hat{H}_x$-colored and
        \item The number of homomorphisms from $\hat{H}_x$ to $\hat{G}_x$ is given by
            \[\#\homs{\hat{H}_x}{\hat{G}_x} = \#\homs{H_x}{G_x} \cdot \#\auts{\K(2\kappa(x)+3)}.\]
    \end{enumerate}
\end{lemma}
\begin{proof} Recall that by definition, we have that $\hat{H}_x := H_x \cup \K(2\kappa(x)+3)$ and
    $\hat{G}_x := G_x \cup \enc(\langle x, H_x \rangle) \cup \K(2\kappa(x)+3)$, where $\enc(\star)$
    encodes a string into a disjoint set of paths and isolated vertices.

    Now, for the homomorphism from $\hat{G}_x$ to $\hat{H}_x$, note that the graph $G_x$ is
    $H_x$-colored (by \cref{lem:parsihard} and in particular \cref{enp1}). Further,
    the graph $\K(2\kappa(x)+3)$ has an automorphism and contains at least one edge, so
    there is a homomorphism from the graph $\enc(\langle x, H_x \rangle) \cup \K(2\kappa(x)+3)$
    into the graph $\K(2\kappa(x)+3)$. In total, this completes the proof that $\hat{G}_x$
    is $\hat{H}_x$-colored.

    For the number of homomorphisms from $\hat{H}_x$ to $\hat{G}_x$, note that there
    are $\#\homs{H_x}{G_x}$ many homomorphisms from ${H}_x$ to ${G}_x$ and
    $\#\auts{\K(2\kappa(x)+3)}$ many homomorphisms from $\K(2\kappa(x)+3)$ to itself.
    As the graphs $\hat{H}_x$ and $\hat{G}_x$ consist of the disjoint union of
    ${H}_x$ and $\K(2\kappa(x)+3)$, and  ${G}_x$ and $\K(2\kappa(x)+3)$, respectively,
    we directly obtain a lower bound for the number of homomorphisms:\begin{equation}
        \#\homs{\hat{H}_x}{\hat{G}_x} \ge \#\homs{H_x}{G_x} \cdot \#\auts{\K(2\kappa(x)+3)}.\label{eqn:lb}
    \end{equation}
    To prove the upper bound, observe the following.
    \begin{enumerate}[i]
        \item The graph $H_x$ cannot be mapped homomorphically to the graph $\K(2\kappa(x)+3)$,
            as $H_x$ is contained in $\mathcal{K}_{even}$ and $\K(2\kappa(x)+3)$ is contained in
            $\mathcal{K}_{odd}$; hence both are distinct Kneser graphs and by
            \cref{lem:kneser_magic} no homomorphisms between them are possible.
        \item The graph $\K(2\kappa(x)+3)$ cannot be mapped homomorphically to the graph $G_x$.
            Suppose otherwise, that an homomorphism $h: \K(2\kappa(x)+3) \to G_x$ existed.
            As $G_x$ is $H_x$-colorable, there is a homomorphism $c: G_x \to H_x$.
            Composing the homomorphisms $h$ and $c$ yields a homomorphism
            \[h\circ c:\K(2\kappa(x)+3) \to H_x,\] that is, a homomorphism from a graph in $\mathcal{K}_{odd}$
            to a graph in $\mathcal{K}_{even}$, which, again, is not possible by
            \cref{lem:kneser_magic}.
        \item None of the graphs $H_x$ and $\K(2\kappa(x)+3)$ can be mapped homomorphically to
            the graph $\enc(\langle x, H_x \rangle)$, as paths have a chromatic number of at most
            $2$, and both graphs $H_x$ and $\K(2\kappa(x)+3)$ have a chromatic number of at least
            $3$.
    \end{enumerate}
    Hence, the homomorphisms counted in the lower bound (\ref{eqn:lb}) are already {\em all}
    homomorphisms from $\hat{H}_x$ to $\hat{G}_x$:\[
        \#\homs{\hat{H}_x}{\hat{G}_x} = \#\homs{H_x}{G_x} \cdot \#\auts{\K(2\kappa(x)+3)}.
    \] This completes the proof.
\end{proof}

Now let $\hat{\mathcal{H}}$ and $\hat{\mathcal{G}}$ be the sets of all graphs $\hat{H}_x$
and $\hat{G}_x$, respectively, corresponding to instances $x$ to $(F,\kappa)$ for which
the function $g$ is non-zero, that is $g(x)\neq 0$.
For the classes $\hat{\mathcal{H}}$ and $\hat{\mathcal{G}}$ to be useful to us, we need to show
that the only pairs of graphs $H\in \hat{\mathcal{H}}$ and $G \in \hat{\mathcal{G}}$ that
admit a homomorphism from $H$ to $G$ are those, that correspond to the same pair $(x,\kappa(x))$.
Formally, consider the following lemma.

\begin{lemma}\label{lem:mainp2}
    Let a problem $(F,\kappa)\in \#\W{1}$ and the corresponding graph classes $\hat{\mathcal{H}}$
    and $\hat{\mathcal{G}}$ be given. For any graphs $K\in\hat{\mathcal{H}}$ and
    $G\in\hat{\mathcal{G}}$, if there is a homomorphism from $H$ to $G$,
    then there is an instance $x$ to $(F,\kappa)$ that corresponds to both
    $H$ and $G$, that is, $H = \hat{H}_x$ and $G = \hat{G}_x$.
\end{lemma}
\begin{proof}
    It suffices to show that there are no homomorphisms from $H$ to $G$ if the graphs
    $H$ and $G$ {\em do not} correspond to the same instance of $(F,\kappa)$.
    Hence, assume that the graphs $K\in \hat{\mathcal{H}}$ and $G\in \hat{\mathcal{G}}$
    correspond to distinct instances $x=x_K$ and $y=x_G$ of $(F,\kappa)$.
    For the sake of contradiction, further assume that there is a homomorphism $h$ from the
    graph $H$ to the graph $G$.
    By \eqref{eqn:k1} and \eqref{eqn:k2}, for some distinct integers $a, b$, an integer $c$,
    and a graph $H_x \not=K(a)$, we have that
    \begin{align*}
        H &= \K(a) \cup \K(b)\text{ and }\\
        G &= G_1 \cup \enc(\langle y, H_y\rangle) \cup \K(c),
    \end{align*} where $G_1$ is a connected graph that is $H_y$-colored.

    Similar to (iii) from the proof of \cref{lem:mainp1} we can show that,
    there are no homomorphisms from the graphs $\K(a)$ or $\K(b)$ to the graph
    $\enc(\langle x, H_x\rangle)$. Further, as the numbers $a$ and $b$ are distinct,
    only at most one of $a$ and $b$ may be equal to $c$. We distinguish two cases,
    depending on whether $c$ is equal to either $a$ or $b$, or not.

    First, assume that the number $b$ is the same as $c$. (The case $a = c$ is similar.)
    In this case, we have that $\K(b) = \K(c)$, and hence $\kappa(x) = \kappa(y)$.
    Further, by \cref{lem:kneser_magic}, the homomorphism $h$ maps the graph $\K(a)$ to
    the graph $G_1$, as $\K(a)$ and $\K(c)$ are different Kneser graphs. Combining this
    homomorphism from $\K(a)$ to $G_1$ with the homomorphism from $G_1$ to $H_y$
    (which exists as $G_1$ is $H_y$-colorable) yields a homomorphism from $\K(a)$ to $H_y$.
    However, as $\K(a)=H_x$ and $H_y$ are both Kneser graphs, a homomorphism between them
    is possible only if they are the same Kneser graph. This in turn, means that
    the instances $x$ and $y$ are the same, which is a contradiction.

    Second, consider the case where the numbers $a$, $b$, and $c$ are pairwise distinct.
    By \cref{lem:kneser_magic}, we obtain that there are no homomorphisms from the
    graph $\K(a)$ to the graph $\K(c)$, as well as that there are no homomorphisms from the
    graph $\K(b)$ to the graph $\K(c)$.
    Hence, the homomorphism $h$ maps both graphs $\K(a)$ and $\K(b)$ to the graph $G_1$.
    Assume wlog. that $\K(a) = H_x$. Now, as before, we obtain a homomorphism from the
    graph~$H_x$ to the graph $H_y$ and hence (by \cref{lem:kneser_magic}) $x=y$,
    which is a contradiction.

    In total, if $H$ and $G$ do not correspond to the same instance $x$, there is
    no homomorphism from $H$ to $G$. This concludes the proof.
\end{proof}

\paragraph*{The Main Reductions}

Using \cref{lem:mainp1,lem:mainp2}, we proceed to show that the problems $(F,\kappa)$ and
$\#\homsprob{\hat{\mathcal{H}}}{\hat{\mathcal{G}}}$ are interreducible with respect to
parameterized Turing reductions.

\begin{proof}[Proof of Theorem~\ref{thm:universal_turing_red}]
    We start with the more involved direction.
    \begin{claim}\label{cl:c2}
        We have that $\#\homsprob{\hat{\mathcal{H}}}{\hat{\mathcal{G}}} \fptred (F,\kappa)$.
    \end{claim}
    \begin{claimproof}
        Given graphs $H$ and $G$ and an oracle $\mathbb{O}$ for the problem $(F,\kappa)$,
        we wish to compute the number of homomorphisms
        $\#\homs{H}{G}$ if the promise $H \in \hat{\mathcal{H}}$ and $G \in \hat{\mathcal{G}}$
        is fulfilled. Consider the following algorithm $\mathbb{B}$.
        \begin{enumerate}
            \item Verify that the graph $H$ is the union of two Kneser graphs
                $\K(a) \in \mathcal{K}_{even}$ and $\K(b) \in \mathcal{K}_{odd}$.
                If this is not the case, output~$0$.
            \item Verify that the graph $G$ is the union of a set of paths $\mathcal{P}$
                (and isolated vertices) and two connected components $G_1$ and $G_2$
                that are not paths. If this is not the case, output~$0$.
            \item Verify that either $G_1=\K(b)$ or $G_2=\K(b)$ holds.
                If this holds, assume w.l.o.g.\ that $G_2 = \K(b)$. Otherwise output~$0$.
            \item Find the pair $(c,H')$ such that $\enc(\langle c,H' \rangle) = \mathcal{P}$
                or report that no decoding is possible (for instance\ if the set of paths $\mathcal{P}$
                is empty, contains the same path multiple times or the number of isolated vertices
                does not match). If the decoding failed, output~$0$.
            \item Compute the parameter $\kappa(c)$ of the instance $c$.
                If we have that $2\kappa(c)+3 \neq b$, output $0$.
            \item Verify that the graphs $H'$ and $\K(a)$ are isomorphic.
                If they are not isomorphic, output $0$.
            \item Compute the value $g(c)$. If we have that $g(c) = 0$, output $0$.
            \item Query the oracle $\mathbb{O}$ on input $c$ and obtain $\mathbb{O}(c)$.
                Output the number \[\mathbb{O}(c) \cdot \#\mathsf{Aut}(\K(2\kappa(c)+3))\cdot g(c)^{-1}\!.\]
        \end{enumerate}
        We first prove the required bound on the running time of the algorithm $\mathbb{B}$.
        On input $H$~and~$G$, Step~1 takes time depending only on $|V(H)|$;
        Step~2 can be done in time polynomial in $|V(G)|$.
        Step~3 takes again time depending only on $|V(H)|$.
        Considering Step~4, we observe that by the definition of the encoding $\mathsf{enc}$
        and by the assumption that $\langle \star,\star \rangle$ is an efficient encoding of
        pairs, the decoding can be done in time polynomial in $|V(\mathcal{P})|\leq |V(G)|$.
        Step~5 can also be done in time polynomial in $|V(G)|$, as the function $\kappa$ is
        computable in polynomial time in $|c|$. As the encoding $\enc(\langle c,H \rangle)$
        contains an isolated vertex for every bit of the string~$c$, we have that
        \begin{equation}
            |c| \leq |V(G)|\label{eqn:231}
        \end{equation}
        and hence the claimed running time for Step~5.
        Similarly to the Step~3, we can perform Step~6 in time depending only on $|V(H)|$.
        Now assume Step~7 is reached. In this case, we have that $2\kappa(c)+3 = b$
        and consequently
        \begin{equation}\label{eq:parbound}
            \kappa(c) \leq  |V(H)|,
        \end{equation}
        as the graph $\K(b)$ is a component of $H$ and we have that $|V(\K(b))| \geq b$.
        Note that Steps~7~and~8 take time $f'(\kappa(c)) \cdot |c|^{O(1)}$,
        as the problem $(g,\kappa)$ is fixed-parameter tractable (see (\ref{enp3})).
        W.l.o.g., we can assume that $f'$ is monotonically increasing and thus~\eqref{eq:parbound}
        yields a running time bound of $f'(|V(H)|) \cdot |V(G)|^{O(1)}$;
        recall \eqref{eqn:231}, that is $|c| \leq |V(G)|$.\\
        Note that the last argument also shows that the parameter $\kappa(c)$ of the oracle query
        $\mathbb{O}(c)$ is bounded by $|V(H)|$.

        It remains to prove the correctness of algorithm $\mathbb{B}$.
        To this end, assume that the promise is fulfilled, that is,
        $H \in \hat{\mathcal{H}}$ and $G \in \hat{\mathcal{G}}$.
        (If the promise is not fulfilled, we are not required to compute a correct output.)

        Hence, for instances $x$ and $y$, we have that\begin{align*}
            H = \hat{H}_y &= \K(a) \cup \K(b)\\
                          &= H_y \cup \K(2\kappa(y)+3)\\
            \intertext{and}
            G=\hat{G}_x   &= G_1 \cup \mathcal{P} \cup G_2\\
                          &= G_x \cup \enc(\langle x, H_x \rangle) \cup \K(2\kappa(x)+3).
        \end{align*} Further, by construction we have that $g(x) \not= 0$.
        We consider three cases.
        \begin{enumerate}[i]
            \item $H_x \neq H_y$: The instances $x$ and $y$ are different. Hence by
                \cref{lem:mainp2}, there are no homomorphisms from $H$ to $G$.

                In the algorithm $\mathbb{B}$, in this case, the test in Step~6 fails,
                and $\mathbb{B}$ outputs $0$, which is correct.
            \item $H_x=H_y$ and $\kappa(x)\neq\kappa(y)$. Note that $H_x=H_y$ does not
                imply that the corresponding instances are the same; the algorithm~%
                $\mathbb{A}$ is not necessarily injective. Indeed, in this case
                the instances $x$ and $y$ differ and so, again by \cref{lem:mainp2},
                there are no homomorphisms from the graph $H$ to the graph $G$.

                In the algorithm $\mathbb{B}$, in this case, the test in Step~3 fails,
                and $\mathbb{B}$ outputs $0$, which is correct.
            \item $H_x=K_y$ and $\kappa(x)=\kappa(y)$: In this case, we have that
                $\hat{H}_y = \hat{H}_x$.
                Hence, by \cref{lem:mainp1}, the number of homomorphisms from $H$ to $G$ is
                \begin{equation}
                    \#\homs{H}{G} = \#\homs{H_x}{G_x} \cdot \#\auts{\K(2\kappa(x) + 3)}.\label{eqn:cl251}
                \end{equation}
                Note that the oracle $\mathbb{O}$ on input $x$ computes the number
                \begin{equation}
                    \mathbb{O}(x) = g(x) \cdot \#\homs{H_x}{G_x}\label{eqn:cl252},
                \end{equation}
                and we may assume that $g(x) \not= 0$ by construction. Hence, combining
                \eqref{eqn:cl251} and \eqref{eqn:cl252} yields that we can compute the number
                of homomorphisms from the graph $H$ to the graph $G$ as follows:\[
                    \#\homs{H}{G} = \mathbb{O}(x)\cdot{g(x)}^{-1}\cdot\#\auts{\K(2\kappa(x) + 3)}.
                \]

                In the algorithm~$\mathbb{B}$, it is easy to verify that the Steps~3~to~7
                succeed and that in Step~8, we indeed return
                $\mathbb{O}(x)\cdot{g(x)}^{-1}\cdot\#\auts{\K(2\kappa(x) + 3)}$.
                Hence, the algorithm is correct in this case as well.
        \end{enumerate}
        In total, the algorithm $\mathbb{B}$ correctly solves the problem
        $\#\homsprob{\hat{H}}{\hat{G}}$. This finishes the proof of the reduction.
    \end{claimproof}

    Finally, we construct and verify the easy reduction.
    \begin{claim}\label{cl:c1}
        We have that $(F,\kappa) \fptred \#\homsprob{\hat{\mathcal{H}}}{\hat{\mathcal{G}}}$.
    \end{claim}
    \begin{claimproof}
        Given an instance $x$ to $(F, \kappa)$ and an oracle $\mathbb{O}$ solving the problem
        $\#\homsprob{\hat{\mathcal{H}}}{\hat{\mathcal{G}}}$, we wish to compute the number $F(x)$.
        Recall that by \cref{lem:parsihard}, there is a reduction
        $(F,\kappa) \fptred \#\homsprob{\mathcal{K}_{even}}{\top}$; let $\mathbb{A}$ again
        denote the corresponding algorithm. Recall further, that for the graphs
        $(H_x, G_x) = \mathbb{A}(x)$, we have that\[
            F(x) = g(x)\cdot \#\homs{H_x}{G_x}.
        \]
        Now, to compute the result $F(x)$, we first compute the value $g(x)$ in FPT time
        with respect to $\kappa$. If we observe $g(x)=0$, we output $0$.
        Otherwise, we simulate the algorithm $\mathbb{A}$ and obtain graphs $H_x$ and $G_x$
        in time $f(\kappa(x))\cdot |x|^{O(1)}\!$.
        After that, we can compute the graphs $\hat{H}_x$ and $\hat{G}_x$
        in time $\hat{f}(\kappa(x)) \cdot |x|^{O(1)}$:
        The construction of the encoding $\enc(\langle x, H_x \rangle)$ can be done in
        polynomial time in $|x|$ and $|V(H_x)|$. Note that $|V(H_x)|$ is bounded by
        $s(\kappa(x))$ and that the construction of $\K(2\kappa(x)+3)$ clearly takes time
        depending only on $\kappa(x)$. In particular, we have that the size of the graph $\hat{H}_x$
        depends only on $\kappa(x)$. Hence, we can query the oracle~$\mathbb{O}$ for the problem
        $\#\homsprob{\hat{\mathcal{H}}}{\hat{\mathcal{G}}}$ on the graphs $\hat{H}_x$ and
        $\hat{G}_x$ and obtain the number of homomorphisms from $\hat{H}_x$ and $\hat{G}_x$.

        Recall that by \cref{lem:mainp1}, we have that
    \begin{align*}
    \#\homs{H_x}{G_x} &= \#\mathsf{Aut}(\K(2\kappa(x)+3))^{-1} \cdot \#\homs{\hat{H}_x}{\hat{G}_x} \\
                      &= \#\mathsf{Aut}(\K(2\kappa(x)+3))^{-1} \cdot\mathbb{O}(\hat{H}_x,\hat{G}_x).\end{align*}
    Hence, to compute the result $F(x)= g(x)\cdot \#\homs{H_x}{G_x}$,
        we can compute the number $\#\mathsf{Aut}(\K(2\kappa(x)+3))^{-1}$ in time
        depending only on
        $\kappa(x)$ and multiply with the result of the oracle query and the result previous
        computation of the value $g(x)$.
        This completes the proof, as we may assume that the algorithm $\mathbb{A}$ is correct.
    \end{claimproof}
    In total, by the reductions from Claims~\ref{cl:c2} and \ref{cl:c1}, we obtain
    \[(F,\kappa) \interred \#\homsprob{\mathcal{H}}{\mathcal{G}},\]
    thus completing the proof.
\end{proof}
Note that the previous proof shows the corresponding theorem for the decision realm,
if we choose the decision version of \cref{lem:parsihard} as a starting point;
the decision version of \Cref{lem:parsihard} can be found as \cref{cor:decision_hard} in
\cref{sec:hom_hard_strategy}.
\begin{theorem}\label{thm:universal_turing_red_dec}
    Let $(F,\kappa)$ be a problem in $\W{1}$. There are classes
    $\mathcal{H}$ and $\mathcal{G}$ such that
    \[(F,\kappa) \interred \homsprob{\mathcal{H}}{\mathcal{G}}.\]
    Further, $\mathcal{H}$ is recursively enumerable and $\mathcal{G}$ is recursive.\lipicsEnd
\end{theorem}

\section{Counting Homomorphisms and Subgraphs in \texorpdfstring{$F$}{F}-Colorable Graphs}\label{sec:fcol}

Let $\mathcal{H}$ denote a recursively enumerable class of graphs. Further, given a fixed graph
$F$, let $\mathcal{G}_F$ denote the class of all graphs $G$ that admit a homomorphism to $F$,
that is, the class of $F$-colorable  graphs.
In this section we establish that the existing dichotomy for counting homomorphisms due to Dalmau and Jonsson~\cite{DalmauJ04}  extends to the PPC problem
$\#\homsprob{\mathcal{H}}{\mathcal{G}_F}$; that is, counting the number of homomorphisms
from a graph $H \in \mathcal{H}$ to a graph $G \in \mathcal{G}_F$.
Note that the notion of $\mathcal{G}_F$ captures and generalizes the important special cases of
the class of all bipartite graphs (when $F$ is a single edge) or, more generally,
the class of all $k$-colorable graphs for any fixed number $k$ (when $F$ is the complete
graph on $k$ vertices).

\begin{theorem}\label{thm:genhomdich}
    Let $F$ be a graph, and let $\mathcal{H}$ be a recursively enumerable class of graphs.
    \begin{enumerate}[~1~]
    \item If the treewidth of $\mathcal{H}\cap\mathcal{G}_F$ is bounded then the PPC problem
        $\#\homsprob{\mathcal{H}}{\mathcal{G}_F}$ is polynomial-time solvable.
    \item Otherwise, the problem $\#\homsprob{\mathcal{H}}{\mathcal{G}_F}$ is $\#\W{1}$-hard.
        \lipicsEnd
    \end{enumerate}
\end{theorem}

It turns out that the previous theorem can be proved by a refined analysis of the existing proof
due to Dalmau and Jonsson~\cite{DalmauJ04}.
For this reason, we defer the proof to \cref{sec:proof_f_homsdicho}.
In what follows, instead, we demonstrate that the previous classification for counting
homomorphisms to $F$-colorable graphs yields a complete classification for the associated subgraph
counting problem. More precisely, we define $\#\subsprob{\mathcal{H}}{\mathcal{G}}$ as the
PPC problem of, given graphs $H \in \mathcal{H}$ and $G \in \mathcal{G}$,
computing the number $\#\subs{H}{G}$, that is, the number of subgraphs in $G$ that are
isomorphic to $H$; the parameter is $|V(H)|$.
Formally, the promise is the set $\mathcal{H}\times\mathcal{G}$.

An example of a problem $\#\subsprob{\mathcal{H}}{\mathcal{G}}$ is the problem of computing
the number of $k$-matchings in bipartite graphs; recall that a $k$-matching is a set of $k$ edges
that are pairwise disjoint.
This problem was first shown to be $\#\W{1}$-hard by Curticapean and Marx~\cite{CurticapeanM14}
and constitutes the bottleneck for the intractable cases of the subgraph counting problem:
\begin{theorem}[\cite{CurticapeanM14}]\label{thm:subsdicho}
    Let $\mathcal{H}$ be a recursively enumerable class of graphs.
    \begin{enumerate}[~1~]
        \item If the matching number of the class $\mathcal{H}$ is bounded then
            the problem $\#\subsprob{\mathcal{H}}{\top}$ is polynomial-time solvable.
        \item Otherwise, the problem $\#\subsprob{\mathcal{H}}{\top}$ is $\#\W{1}$-hard.
            \lipicsEnd
    \end{enumerate}
\end{theorem}

Here, the {\em matching number} of a graph is the size of its largest matching
and a class of graphs $\mathcal{H}$ has bounded matching number if there is an overall
constant $c$ such that the matching number of each graph $H \in \mathcal{H}$ is bounded by $c$.

Recently, Curticapean, Dell, and Marx~\cite{CurticapeanDM17} strongly generalized \cref{thm:subsdicho} with a much simpler proof.
They key ingredient of their work is the algorithm given by \cref{lem:monotonicity}.
We further generalize their proof to $F$-colorable graphs and obtain the following strengthening
of the classification for counting subgraphs.

\begin{theorem}\label{thm:subs_f_colored}
    Let $F$ be a fixed graph and let $\mathcal{H}$ be a recursively enumerable class of graphs.
    \begin{enumerate}[~1~]
        \item If the matching number of $\mathcal{H}\cap \mathcal{G}_F$ is bounded then
            the problem $\#\subsprob{\mathcal{H}}{\mathcal{G}_F}$ is polynomial-time solvable.
        \item Otherwise, the problem $\#\subsprob{\mathcal{H}}{\mathcal{G}_F}$ is $\#\W{1}$-hard.
            \lipicsEnd
    \end{enumerate}
\end{theorem}

Due to space constraints and the fact that we need to perform only minor modifications of the
arguments of Curticapean, Dell and Marx~\cite{CurticapeanDM17}, we defer the proof to
\cref{sec:app_fcol_sub}.

\section{Counting Homomorphisms in Kőnig Graphs}\label{sec:line_graphs}

Given a graph $G$, its associated {\em line graph} $L(G)$ is the following graph:
As vertices $L(G)$ has the edges of $G$ and two vertices $e$ and $\hat{e}$ of $L(G)$ are adjacent
if the corresponding edges are neither equal nor disjoint, that is, $|e \cap\hat{e}|=1$.
We write $\mathcal{L}$ for the set of all line graphs.
By Kőnig's Theorem, a line graph of a bipartite graph is also a perfect
graph (see for instance\ \cite{ChudnovskyRST06}). To simplify notation, we hence call
a line graph of a bipartite graph a {\em Kőnig graph}. We write $\symking$ to denote the
class of all Kőnig graphs\footnote{The symbol $\symking$ is used since ``Kőnig'' is the German
word for ``King''.}. This section is devoted to the complexity analysis of the
problem $\#\homsprob{\mathcal{H}}{\symking}$ of counting homomorphisms from a graph from some
arbitrary graph class $\mathcal{H}$ to a Kőnig graph.

\subsection{Tractability of Deciding Homomorphisms in Kőnig Graphs}

We start by investigating the decision version, that is, the problem $\homsprob{\top}{\symking}$.
It turns out that if we are interested only in the existence, and not the number, of homomorphisms,
then the problem becomes fixed-parameter tractable.

\begin{theorem}\label{thm:dec_lines_tract}
    The decision problems $\homsprob{\top}{\mathcal{L}}$ and thus $\homsprob{\top}{\symking}$
    are fixed-parameter tractable.
    In particular, given a graph $H$ and a line graph $L$, it is possible to decide the existence
    of a homomorphism from $H$ to $L$ in time \[f(|V(H)|) \cdot O(|V(L)|^2),\] for some computable
    function $f$ independent of $H$ and $L$.
\end{theorem}
\begin{proof}
    We construct an algorithm $\mathbb{A}$ for the problem $\homs{\top}{\mathcal{L}}$ that,
    given a graphs $H \in \top$ and $L \in \mathcal{L}$, correctly decides whether there
    is a
    homomorphism from $H$ to $L$. Further, the algorithm $\mathbb{A}$ runs in time
    $f(|V(H)|) \cdot O(|V(L)|^2)$ for some computable function $f$ independent of $H$ and $L$.
    The algorithm $\mathbb{A}$ relies on the clique partition of line
    graphs~\cite[Chapter~8]{Harary69}, stating that $E(L)$ can be partitioned into cliques such
    that every vertex of $L$ is contained in at most $2$ cliques.
    Here, every clique corresponds to a vertex of a primal graph of $G$ such that $L(G)=L$.
    In particular, it is easy to see that the size of the largest clique in the partition is
    precisely the maximum degree of $G$. Consequently, our algorithm first computes a primal graph
    $G$ of $L$, which can be done in time $O(|V(L)|^2)$~\cite{Lehot74}.
    Next, we compute the maximum degree $d$ of $G$, which can be done in time
    $O(|V(G)|) = O(|V(L)|^2)$.

    Now let $k =|V(H)|$ and let $H_1,\dots, H_\ell$ be the connected components of $H$.
    For every connected component $H_i$, we proceed as follows.
    If $d \geq k$ then there is a homomorphism from $H_i$ to $L$, as we can embed $H_i$
    into a clique of size $d$.
    Otherwise, the properties of the clique partition yield that the degree of $L$ is bounded
    by $2k$: Every vertex of $L$ is contained in at most two cliques and every clique is of size
    at most $d < k$.
    Consequently, we can perform a standard bounded search-tree algorithm:
    We guess the image $v \in V(L)$ of a vertex $h\in V(H_i)$.
    As the graph $H_i$ is connected and $|V(H_i)| \leq k$, every homomorphism from $H_i$ to $L$
    that maps $h$ to $v$ must also map every further vertex of $H_i$ to a vertex in the
    $k$-neighborhood of $v$. As the maximum degree of $L$ is at most $2k$, the size of the graph
    induced by the $k$-neighborhood of~$v$ is bounded by $(2k)^k$.
    We can then search for a homomorphism by brute-force; this takes time depending only on $k$.
    The final output is $1$ if a homomorphism is found from every connected component $H_i$
    and $0$ otherwise.

    The total running time is bounded by
    \[O(|V(L)|^2) + f(|V(H)|) \cdot O(|V(L)|) \leq f(|V(H)|) \cdot O(|V(L)|^2);\]
    this completes the proof.
\end{proof}

\subsection{An Explicit Criterion for Hardness of Counting Homomorphisms in Kőnig Graphs}

\Cref{thm:dec_lines_tract} in turn further motivates the study of the counting version:
The most interesting hardness results in counting complexity theory are concerned with problems
that admit a tractable decision version~\cite{Valiant79}.
In particular, we construct an explicit reduction from $\#\clique$ to prove the following
hardness result.
\begin{lemma}\label{lem:main_line_graphs}
    Let $\mathcal{H}$ be a recursively enumerable class of graphs.
    If $\mathcal{H}$ has unbounded treewidth and is closed under taking minors, then
    the problem $\#\homsprob{\mathcal{H}}{\symking}$ is $\#\W{1}$-hard.
\end{lemma}
Note that Kőnig graphs are a subset of the perfect graphs~\cite{ChudnovskyRST06},
as well as a subset of the line graphs (of arbitrary graphs). As line graphs are also
claw-free graphs~\cite{Beineke70}, Kőnig graphs are also a subset of the claw-free graphs.
Hence, the hardness result for the problem $\#\homsprob{\mathcal{H}}{\symking}$
extends to perfect graphs, line graphs, and claw-free graphs as well.

To prove \cref{lem:main_line_graphs}, we use a gadget construction that transforms an arbitrary
graph $G$ into a Kőnig graph such that the number of grid-like subgraphs remains stable.
In view of the diverse applications of the Grid-Tiling Problem
(see for instance\ \cite[Chapter~14.4.1]{CyganFKLMPPS15}), the construction might yield further
intractability results for counting problems on Kőnig graphs (and hence on claw-free and perfect graphs).

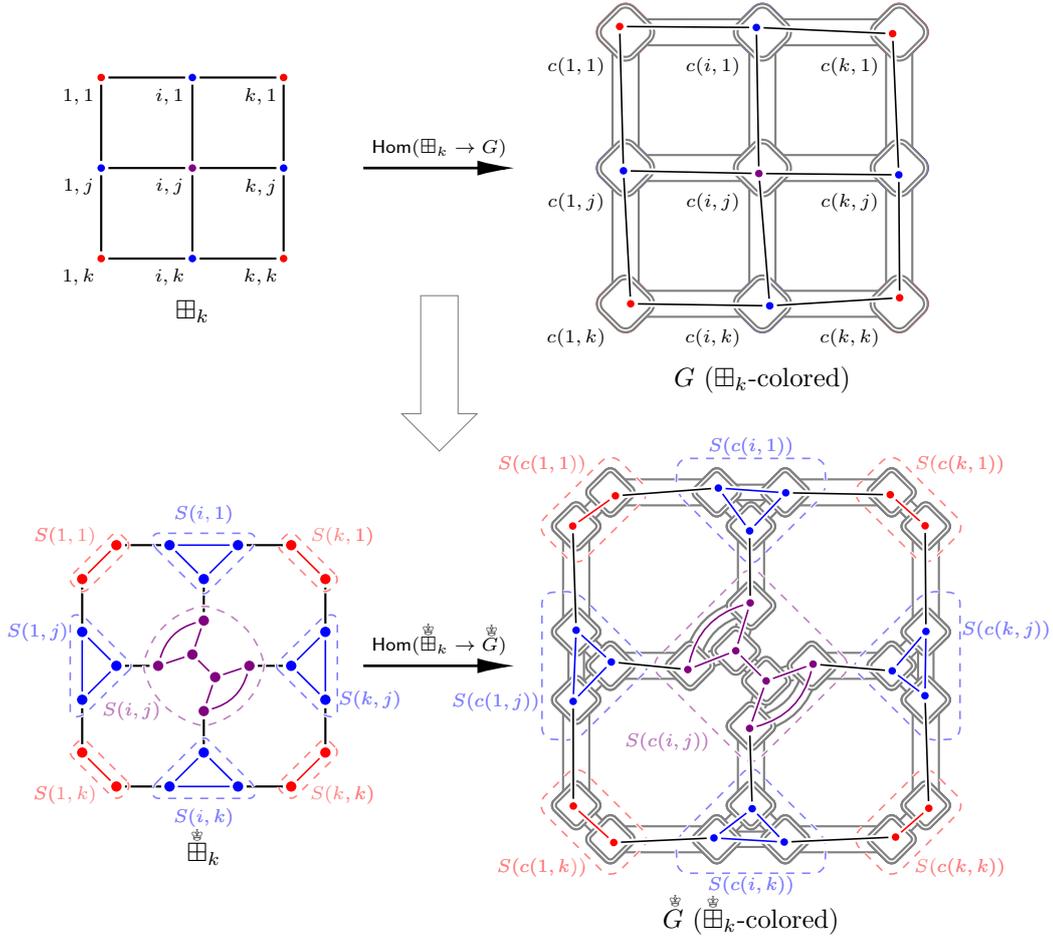
\begin{figure}[t]
    \centering
    \begin{tikzpicture}
        \pic at (0,-.4) {grid=3};
        \pic[kneser style=1] at (7.5,-.4) {grid=3};
        \draw[line width=1pt,-{Latex[length=5mm, width=2mm]}] (2.25,-0.4)
            --node[above,midway,sloped]
            {\scriptsize$\homs{\boxplus_k}{G}$} (4.25,-0.4);
        \node at (0,-2.3) {$\boxplus_k$};
        \node at (7.5,-3.2) {$G$ ($\boxplus_k$-colored)};

        \node[single arrow, rotate=-90, draw, gray] at (3.25,-3) {\color{white}{Reduction}};
        \pic[kneser style=20] at (0.15,-7) {grid=3};
    \pic[kneser style=21] at (7.35,-7) {grid=3};
\draw[line width=1pt,-{Latex[length=5mm, width=2mm]}] (2.25,-7) --node[above,midway,sloped]
    {\scriptsize$\homs{\xcrown{\boxplus}_k}{\xcrown{G}}$} (4.25,-7);
\node at (0.15,-9.4) {$\xcrown{\boxplus}_k$};
\node at (7.35,-10.3) {$\xcrown{G}$ (``$\xcrown{\boxplus}_k$-colored'')};
    \end{tikzpicture}

    \caption{General overview of the reduction from \cref{lem:main_line_graphs}:
        Corner vertices (red), border vertices (blue), and interior vertices (purple)
        get replaced by corresponding gadgets; the resulting graph $\xcrown{G}$ is a Kőnig graph.
        We use $c(\star)$ to denote the set of all vertices colored with $\star$ and
        $S(\star)$ to denote the gadget corresponding to $\star$; the numbers $i$ and $j$
        denote intermediate columns and rows.
    }\label{fig:genred}
\end{figure}%
\begin{figure}[t]
    \centering
    \begin{tikzpicture}
        \node[vertex,fill=\ccol!50] at (-3.85,3) (c10) {};
        \node[vertex,fill=\ccol] at (-4.15,3) (c0) {};
        \node[outer sep=.7pt, inner sep=0] at (-3.85,3) (10) {};
        \node[outer sep=.7pt, inner sep=0] at (-4.15,3) (0) {};
        \node[fit=(0)(10),draw=gray,double=white,
            rounded corners=5pt]{};
        \node[] at (-4.65,3.15) (c1) {};
        \node[] at (-4.65,2.85) (c2) {};
        \node[] at (-4.15,2.35) (c3) {};
        \node[] at (-4.65,3.45) (c11) {};
        \node[] at (-4,2.35) (c12) {};
        \node[] at (-3.7,2.35) (c13) {};
        \draw[white, double=gray] (c0) -- (c1);
        \draw[white, double=gray] (c0) -- (c2);
        \draw[white, double=gray] (c0) -- (c3);
        \draw[white, double=gray] (c10) -- (c11);
        \draw[white, double=gray] (c10) -- (c12);
        \draw[white, double=gray] (c10) -- (c13);
        \node at (-3.25,2.8) {\scriptsize $c(k,1)$};

        \begin{scope}[shift={(.5,0)}]
            \node[vertex,fill=\ecol!50] at (.15,3) (c10) {};
            \node[vertex,fill=\ecol] at (-.15,3) (c0) {};
            \node[outer sep=.7pt, inner sep=0] at (-.15,3) (10) {};
            \node[outer sep=.7pt, inner sep=0] at (.15,3) (0) {};
            \node[fit=(0)(10),draw=gray,double=white,
                rounded corners=5pt]{};
            \node[] at (-.65,2.85) (c2) {};
            \node[] at (-.65,3.15) (c1) {};
            \node[] at (-.65,3.45) (c11) {};
            \node[] at (.65,3.15) (d2) {};
            \node[] at (.65,3.45) (d1) {};
            \node[] at (-.3,2.35) (c3) {};
            \node[] at (0,2.35) (c12) {};
            \node[] at (.3,2.35) (c13) {};
            \draw[white, double=gray] (c0) -- (c1);
            \draw[white, double=gray] (c0) -- (c2);
            \draw[white, double=gray] (c0) -- (c3);
            \draw[white, double=gray] (c10) -- (c11);
            \draw[white, double=gray] (c10) -- (d1);
            \draw[white, double=gray] (c10) -- (d2);
            \draw[white, double=gray] (c0) -- (c12);
            \draw[white, double=gray] (c10) -- (c13);
            \node at (.75,2.8) {\scriptsize $c(i,1)$};
        \end{scope}

        \begin{scope}[shift={(5.6,0)}]
            \node[vertex,fill=\icol!50] at (.15,3) (c10) {};
            \node[vertex,fill=\icol] at (-.15,3) (c0) {};
            \node[outer sep=.7pt, inner sep=0] at (-.15,3) (10) {};
            \node[outer sep=.7pt, inner sep=0] at (.15,3) (0) {};
            \node[fit=(0)(10),draw=gray,double=white,
                rounded corners=5pt]{};
            \node[] at (-.65,3.2) (c2) {};
            \node[] at (-.65,3.4) (c1) {};
            \node[] at (-.65,3.6) (c11) {};
            \node[] at (.65,3.2) (d2) {};
            \node[] at (.65,3.4) (d1) {};
            \node[] at (.65,3.6) (d3) {};
            \node[] at (-.15,2.35) (c3) {};
            \node[] at (.15,2.35) (c12) {};
            \node[] at (-.15,3.65) (e3) {};
            \node[] at (.15,3.65) (e12) {};
            \draw[white, double=gray] (c0) -- (c2);
            \draw[white, double=gray] (c0) -- (c3);
            \draw[white, double=gray] (c0) -- (e3);
            \draw[white, double=gray] (c0) -- (d1);
            \draw[white, double=gray] (c0) -- (d3);
            \draw[white, double=gray] (c10) -- (d2);
            \draw[white, double=gray] (c10) -- (c1);
            \draw[white, double=gray] (c10) -- (c11);
            \draw[white, double=gray] (c10) -- (c12);
            \draw[white, double=gray] (c10) -- (e12);
            \node at (.75,2.8) {\scriptsize $c(i,j)$};
        \end{scope}

        \node[single arrow, rotate=-90, draw,gray] at (-4,2) {\color{white}{R}};
        \node[single arrow, rotate=-90, draw,gray] at (.5,2) {\color{white}{R}};
        \node[single arrow, rotate=-90, draw,gray] at (5.6,2) {\color{white}{R}};

        \pic[kneser style=30] at (-4.4,-1) {corner=3/3/3};
        \pic[kneser style=30] at (0,-1) {edge=1/3/3};
        \pic[kneser style=30] at (5.6,-.5) {inner=2/2/3};
    \end{tikzpicture}
    \caption{The gadgets in detail.}\label{fig:gad}
\end{figure}
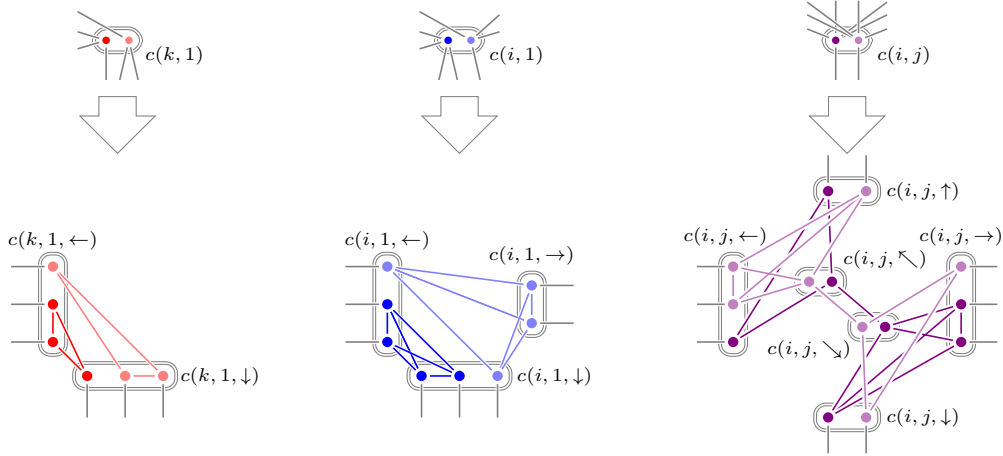

\begin{proof}
    We write $\boxplus_k$ for the $k \times k$ square grid, that is, the graph with the vertices\[
        V(\boxplus_k) := \{(i,j)\mid i,j \in [k] \},\]
    and two vertices $(i,j)$ and $(i',j')$ are adjacent if $|i-i'|+|j-j'| = 1$.
    Now let $\boxplus$ be the set of all square grids $\boxplus_k$ for $k \in \mathbb{N}$.
    We prove a reduction from the problem $\#\cphomsprob{\boxplus}{\top}$, which is known to be
    $\#\W{1}$-hard and constitutes an important intermediate step in the proof of the
    classification of the homomorphism counting problem due to Dalmau and
    Jonsson~\cite{DalmauJ04}. A sketch of the $\#\W{1}$-hardness proof can be found in
    \cref{sec:hom_hard_strategy} and the full proof can be found
    for instance\ in~\cite[Lemma~5.7]{Curticapean15}.

    Let us recall the definition of the problem $\#\cphomsprob{\boxplus}{\top}$.
    This problem expects as input a pair of a square grid $\boxplus_k$ and a graph $G$ that
    is $\boxplus_k$-colored by some given coloring~$c$.
    The task is to compute the number of color-prescribed homomorphisms from $\boxplus_k$ to $G$,
    that is, homomorphisms $h \in \homs{\boxplus_k}{G}$ that additionally satisfy $c(h(v))=v$
    for every vertex~$v$ of the grid.

    For the first part of the reduction, we present a construction that maps a
    $\boxplus_k$-colored graph~$G$ to a vertex-colored Kőnig graph $\xcrown{G}$;
    consider \cref{fig:genred} for an overview of the construction.
    Let $c$ be the coloring of $G$. We partition the vertices of $G$ into three disjoint sets
    (again, consider \cref{fig:genred}):
    \begin{enumerate}[~1~]
        \item A vertex $v$ is called a \emph{corner vertex} if its coloring $c(v)$
            is one of the values $(1,1)$, $(1,k)$, $(k,1)$, and $(k,k)$.
        \item A vertex $v$ is called a \emph{border vertex} if its coloring $c(v)$
            satisfies\[
                c(v)\in \{(i,j)\in[k]^2\mid i \in \{1,k\} \vee j \in\{1,k\} \}
                \setminus \{(1,1),(1,k),(k,1),(k,k)\}.\]
        \item All remaining vertices are called \emph{interior vertices}.
    \end{enumerate}
    We construct a gadget (graph) $S(v)$ for each vertex $v \in V(G)$.
    Here, the graph $S(v)$ depends on whether $v$ is a corner, a border or an interior vertex.
    \begin{enumerate}[~1~]
        \item The vertex $v$ is a corner vertex.
            Assume that $c(v)=(1,1)$; the other cases are symmetric.
            Now let $N_{\rightarrow}$ be the set of neighbors of $v$ that are colored by $c$
            with $(1,2)$ and let $N_{\downarrow}$ be the set of all neighbors of $v$ that are
            colored by $c$ with $(2,1)$.
            Note that that those are all neighbors of $v$ as $G$ is $\boxplus_k$-colored.
            For each vertex $u \in N_{\rightarrow}$, we add a vertex $v^u_\rightarrow$ and color
            it with $(1,1,\rightarrow)$. For each vertex $u \in N_{\downarrow}$,
            we add a vertex $v^u_\downarrow$ and color it with $(1,1,\downarrow)$.
            The graph $S(v)$ is then obtained by making all of the previous vertices adjacent
            to each other.
        \item The vertex $v$ is a border vertex.
            Assume that $c(v)=(1,j)$; the other cases are symmetric.
            Now let $N_{\rightarrow}$ be the set of neighbors of $v$ that are colored by $c$
            with $(1,j+1)$, let $N_{\leftarrow}$ be the set of neighbors of $v$ that are colored
            by $c$ with $(1,j-1)$, and let $N_{\downarrow}$ be the set of all neighbors of $v$ that
            are colored by $c$ with $(2,j)$.
            For each vertex $u \in N_{\rightarrow}$, we add a vertex $v^u_\rightarrow$
            and color it with $(1,i,\rightarrow)$.
            We proceed similarly with $N_{\leftarrow}$ and $N_{\downarrow}$.
            The graph $S(v)$ is then obtained by making all of the previous vertices adjacent.
        \item The vertex $v$ is an interior vertex.
            Let $v$ have color $c(v)=(i,j)$ and let $N_{\rightarrow}$, $N_{\leftarrow}$,
            $N_{\uparrow}$ and $N_{\downarrow}$ be the sets of neighbors of $v$ that are colored
            by $c$ with $(i,j+1)$, $(i,j-1)$, $(i-1,j)$ and $(i+1,j)$, respectively.
            We add a new vertex $v^u_\rightarrow$ for each vertex $u \in N_{\rightarrow}$
            and color it with $(i,j,\rightarrow)$; we proceed similarly with the sets
            $N_\uparrow,N_\leftarrow$, and $N_\downarrow$.
            Next, we add two new vertices $v^\nwarrow$ and $v^{\searrow}$,
            color them $(i,j,\nwarrow)$ and $(i,j,\searrow)$, and connect them by an edge.
            Then we create two cliques: The first clique contains the vertex $v^\star$ and all
            vertices that we colored with $(i,j,\leftarrow)$ or with $(i,j,\uparrow)$.
            The second clique contains the vertex $v^\searrow$ and all vertices that we colored
            with $(i,j,\rightarrow)$ or with $(i,j,\downarrow)$. The resulting graph is $S(v)$.
    \end{enumerate}
    The graph $\xcrown{G}$ is obtained by connecting the gadgets as follows:
    Let $\{v,w\} \in E(G)$ denote an edge of $G$ and assume that the vertex $v$ has color
    $c(v)=(i,j)$ and the vertex $w$ has color $c(w)=(i,j+1)$; the remaining cases are processed
    similarly. By construction, the graph $S(v)$ contains a vertex $v^{w}_\rightarrow$ and
    the graph $S(w)$ contains a vertex $w^{v}_\leftarrow$. We connect those two vertices with an
    edge.

    We first observe that this construction yields a planar graph if it is applied to the grid
    itself (Again, consider \cref{fig:genred}). Further, when applied to the graph $G$,
    we indeed obtain a Kőnig graph:
    \begin{claim}
        If the graph $G$ is $\boxplus_k$-colored, then the graph $\xcrown{G}$ is a Kőnig graph.
    \end{claim}
    \begin{claimproof}
        We construct a bipartite graph $B$ such that the line graph $L(B)$ of $B$ is
        the graph $\xcrown{G}$. To this end, we observe that the gadgets of corner and border
        vertices are cliques, and the gadgets of interior vertices are two cliques that are
        connected by a single edge. Hence, the entire graph $\xcrown{G}$ is obtained by connecting
        vertex disjoint cliques with edges that are pairwise disjoint.

        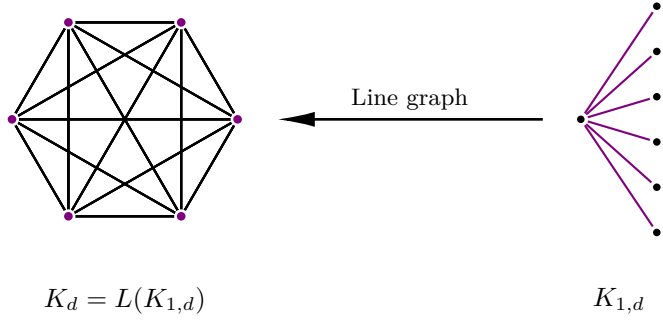
\begin{figure}
            \centering
            \begin{tikzpicture}
                \pgfmathsetmacro{\d}{6}

                \pic[kneser style=13,vertex fill color=\icol] at (-\d,0) {kneser=\d/1};

                \node at (-\d,-\d/2.5) {$K_d = L(K_{1,d})$};

                \draw[line width=1pt,-{Latex[length=5mm, width=2mm]}]
                (-.5,0) --node[above,midway] {\small Line graph} (-\d/1.5,0);

                \node[vertex] (a) at (0,0) {};
                \foreach\i in {1,...,\d}{
                    \node[vertex] (b\d) at (1,-.5*.6+-.6*\d/2+\i*.6) {};
                    \draw[thick,\icol] (b\d) -- (a);
                }
                \node at (.5,-\d/2.5) {$K_{1,d}$};

            \end{tikzpicture}
            \caption{A clique $K_d$ of size $d$ is the line graph of a star $K_{1,d}$ with
            $d$ rays.}\label{fig:stln}
        \end{figure}

        Now observe that cliques are the line graphs of stars.
        More precisely, let $K_{1,d}$ be the complete bipartite graph with $1$ vertex on the left
        side and $d$ vertices on the right side, then its line graph $L(K_{1,d})$ is the clique of
        size $d$; consider \cref{fig:stln} for a visualization.
        In what follows, we say that the single vertex on the left side of $K_{1,d}$ is the
        \emph{center} and the $d$ vertices on the right side are the \emph{rays}.

        Now, adding an edge between two vertex disjoint cliques
        corresponds to merging the right vertices of the corresponding rays of the primal graphs
        (Consider \cref{fig:stln2} for a visualization).
        Consequently, we can construct a graph $B$ whose line graph is $\xcrown{G}$ by
        merging right vertices of the rays corresponding to the edges that connect the cliques
        of the gadgets.

        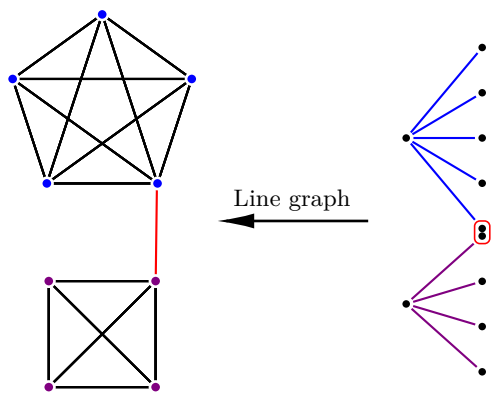
\begin{figure}
            \centering
            \begin{tikzpicture}
                \pgfmathsetmacro{\d}{6}

                \draw[thick,\ccol] (-5.295,-.73) -- (-5.28,.46);
                \pic[kneser style=13,vertex fill color=\ecol] at (-6,1.5) {kneser=5/1};
                \pic[kneser style=13,vertex fill color=\icol] at (-6,-1.5) {kneser=4/1};

                \node[vertex] (a) at (-2,1.1) {};
                \node[vertex] (c) at (-2,-1.1) {};
                \foreach\i in {1,...,5}{
                    \node[vertex] (b\i) at (-1,-.5*.6+-.6*5/2+\i*.6+1.1) {};
                    \draw[thick,\ecol] (b\i) -- (a);
                    \ifthenelse{\i<5}{
                        \node[vertex] (d\i) at (-1,-.5*.6+-.6*4/2+\i*.6-1.1) {};
                        \draw[thick,\icol] (d\i) -- (c);
                    }{}
                }
                \draw[white,double=\ccol,thick,rounded corners=2pt] ($(b1)+(-.1,-.2)$) rectangle ($(d4)+(.1,.2)$);

                \draw[line width=1pt,-{Latex[length=5mm, width=2mm]}]
                (-2.5,0) --node[above,midway] {\small Line graph} (-4.5,0);

            \end{tikzpicture}
            \caption{Connecting two vertex disjoint cliques corresponds to merging
            vertices in the corresponding primal graphs. Note that the resulting primal
            graph stays bipartite.}\label{fig:stln2}
        \end{figure}

        Finally, it is easy to see that $B$ is bipartite: A $2$-coloring is given by the function
        that maps the centers to $1$ and the (identifications of) rays to $2$.
    \end{claimproof}

    Now, recall that we colored the vertices of $\xcrown{G}$ with triples $(i,j,\star)$,
    where $\star$ is one of the symbols $\rightarrow,\leftarrow,\uparrow,\downarrow,\nwarrow$,
    and~$\searrow$.
    Call this mapping \(\xcrown{c}\) and observe
    that \(\xcrown{c}\) is in fact not a
    $\xcrown{\boxplus}_k$-coloring of $\xcrown{G}$,
    as two vertices of $G$ of the same color are adjacent in the gadget
    construction.

    Hence, define a \emph{weakly color-prescribed} homomorphism
    as a \(h \in \homs{\xcrown{\boxplus}_k}{\xcrown{G}}\) such that
    the mapping \(h \circ\xcrown{c}\) is the identity.
    We abuse notation and write
    \(\cphoms{\xcrown{\boxplus}_k}{\xcrown{G}}\) for the set of all weakly
    color-prescribed homomorphisms.

    \begin{claim}
        We have \(\#\cphoms{\xcrown{\boxplus}_k}{\xcrown{G}}
        = \#\cphoms{\boxplus_k}{G}\).
    \end{claim}
    \begin{claimproof}
        Let $h \in \cphoms{\boxplus_k}{G}$ denote a color-prescribed homomorphism.
        We define the mapping $\xcrown{h} \in
        \cphoms{\xcrown{\boxplus}_k}{\xcrown{G}}$
        as follows: For every $i,j \in [k]$ we set
        \begin{alignat*}{3}
            \xcrown{h}(i,j,\rightarrow) & :=
            h(i,j)^{h(i,j+1)}_{\rightarrow};\quad &
            \xcrown{h}(i,j,\leftarrow)  & := h(i,j)^{h(i,j-1)}_{\leftarrow}  \\
            \xcrown{h}(i,j,\downarrow)  & :=
            h(i,j)^{h(i+1,j)}_{\downarrow};
                                        &
            \xcrown{h}(i,j,\uparrow)    & := h(i,j)^{h(i-1,j)}_{\uparrow}    \\
            \xcrown{h}(i,j,\searrow)    & := h(i,j)^\searrow;
                                        &
            \xcrown{h}(i,j,\nwarrow)    & := h(i,j)^\nwarrow
        \end{alignat*}
        The construction of $\xcrown{G}$ immediately yields that $\xcrown{h}$ is a
        (weakly color-prescribed)
        homomorphism if~$h$~is color-prescribed.
        Further, the mapping $h \mapsto \xcrown{h}$ is a bijection.
    \end{claimproof}

    Now let $h$ be a homomorphism from $\xcrown{\boxplus}_k$ to $\xcrown{G}$
    for some $\boxplus_k$-colored graph $G$.
    We call $h$ \emph{colorful} if for every color (or vertex) $(i,j,\star)$ in
    $V(\xcrown{\boxplus}_k)$, there is a vertex in the image of $h$ that is colored with
    $(i,j,\star)$. Denote the set of all colorful homomorphisms from $\xcrown{\boxplus}_k$ to
    $\xcrown{G}$ with $\cfhoms{\xcrown{\boxplus}_k}{\xcrown{G}}$.

    \begin{claim}
        We have $\#\cfhoms{\xcrown{\boxplus}_k}{\xcrown{G}} = \#\auts{\xcrown{\boxplus}_k} \cdot \#\cphoms{\xcrown{\boxplus}_k}{\xcrown{G}}$.
    \end{claim}
    \begin{claimproof}
        Say that two colorful homomorphisms are equivalent if their image is equal.
        Then, every equivalence class has size $\#\auts{\xcrown{\boxplus}_k}$ and is represented
        by a homomorphism for which the induced automorphism is the identity.
        This homomorphism is then not only colorful but also (weakly) color-prescribed.
    \end{claimproof}

    \begin{claim}
        We can compute \(\#\cfhoms{\xcrown{\boxplus}_k}{\xcrown{G}}\)
        in time \(2^{|V(\xcrown{\boxplus}_k)|} \cdot |V(\xcrown{G})|^{O(1)}\)
        by querying the oracle for $\#\homsprob{\mathcal{H}}{\symking}$.
        Further, every oracle query $(\hat{H},\hat{G})$ satisfies that the size~$|V(\hat{H})|$
        depends only on the size $|V(\xcrown{\boxplus}_k)|$.
    \end{claim}
    \begin{claimproof}
        By the principle of ``inclusion and exclusion'', we have that the number of
        colorful homomorphisms can be computed as
        \begin{equation}\label{eq:king_incl_excl}
            \#\cfhoms{\xcrown{\boxplus}_k}{\xcrown{G}} =
            \sum_{J\subseteq V(\xcrown{\boxplus}_k)} (-1)^{|J|} \cdot
            \#\homs{\xcrown{\boxplus}_k}{\xcrown{G}\setminus J},
        \end{equation}
        where $\xcrown{G}\setminus J$ is the graph obtained from $\xcrown{G}$ by deleting all
        vertices that are colored by $\hat{c}$ with a color in $J$.
        Hence we can compute, using the previous equation, the number of colorful homomorphisms
        in time $2^{|V(\xcrown{\boxplus}_k)|} \cdot |V(\xcrown{G})|^{O(1)}$ if an oracle to the
        function $A \mapsto \#\homs{\xcrown{\boxplus}_k}{A}$ is provided.
        In particular, Kőnig graphs are closed under the removal of
        vertices: Deleting a vertex in a Kőnig graph is equivalent to deleting an edge in primal bipartite graph and bipartiteness is closed under the removal of edges. It hence suffices to restrict the graphs $A$ to the class
        $\symking$.

        Now recall that the graph $\xcrown{\boxplus}_k$ is planar.
        By the Excluded Grid Theorem~\cite{RobertsonS86-ExGrid}, every class $\mathcal{H}$ of
        unbounded treewidth contains arbitrary large grids as minors.
        Further, every planar graph is the minor of some grid~\cite{RobertsonST94}.
        As the class $\mathcal{H}$ is minor-closed, we hence obtain that $\xcrown{\boxplus}_k$
        is contained in $\mathcal{H}$ for every $k \in \mathbb{N}$. Consequently, we can
        compute~\eqref{eq:king_incl_excl} using the given oracle for
        $\#\homsprob{\mathcal{H}}{\symking}$.
    \end{claimproof}

    Combining the previous claims, we obtain the claimed result.
\end{proof}

Now, consider the following application of \cref{lem:main_line_graphs}.
\begin{theorem}\label{cor:minor_closed_classification}
    Let $\mathcal{C}$ be one of the classes of line-graphs, claw-free graphs or perfect graphs,
    or a non-empty union thereof.
    Further, let $\mathcal{H}$ be a recursively enumerable class of graphs.
    \begin{enumerate}[~1~]
        \item If the treewidth of the class $\mathcal{H}$ is bounded, then the problem
            $\#\homsprob{\mathcal{H}}{\mathcal{C}}$ is solvable in polynomial time.
        \item Otherwise, if the class $\mathcal{H}$ is additionally minor-closed, the problem
            $\#\homsprob{\mathcal{H}}{\mathcal{C}}$ is $\#\W{1}$-hard.
    \end{enumerate}
\end{theorem}
\begin{proof}
    We immediately obtain the reduction
    $\#\homsprob{\mathcal{H}}{\mathcal{C}} \fptred \#\homsprob{\mathcal{H}}{\top}$.
    In particular, the reduction is the identity and preserves not only fixed-parameter
    tractability, but also polynomial-time tractability.

    By the classification of Dalmau and Jonsson~\cite{DalmauJ04},
    the problem $\#\homsprob{\mathcal{H}}{\top}$ is solvable in polynomial time if
    the class $\mathcal{H}$ has bounded treewidth.
    If the class $\mathcal{H}$ has unbounded treewidth, $\#\W{1}$-hardness follows from
    \cref{lem:main_line_graphs}, as the set of Kőnig graphs is a subset of claw-free
    graphs~\cite{Beineke70}, a subset of perfect graphs~\cite{ChudnovskyRST06} and, of course,
    a subset of line graphs.
\end{proof}

\subsection{An Implicit Criterion for Hardness of Counting Homomorphisms in Kőnig Graphs}

We complement the explicit criterion for hardness for the problem $\#\homsprob{\mathcal{H}}{\mathcal{\symking}}$
from \cref{cor:minor_closed_classification} (which  works only if the class $\mathcal{H}$ is closed under
taking minors) with the following implicit exhaustive complexity classification.

\begin{theorem}\label{thm:implicit_line_dicho}
    Let $\mathcal{H}$ be a recursively enumerable class of graphs.
    Then the problem $\#\homsprob{\mathcal{H}}{\symking}$ is either fixed-parameter tractable
    or $\#\W{1}$-hard under parameterized Turing-reductions.\lipicsEnd
\end{theorem}

In particular, \cref{thm:implicit_line_dicho} shows that the negative result from
\cref{sec:ladnersection} does not apply to Kőnig graphs.

A central ingredient for the proof of \cref{thm:implicit_line_dicho} is the following lemma.
\begin{lemma}\label{lem:qgraphs_linegraphs}
    Let $H$ be a graph. There is a quantum graph $Q[H]$ such that we have for any bipartite
    graph $G$:
    \begin{equation}\label{eq:line_graph_a}
        \#\homs{H}{L(G)} = \#\homs{Q[H]}{G}.
    \end{equation}
    In particular, the mapping $H \mapsto Q[H]$ is computable.\lipicsEnd
\end{lemma}
The proof of \cref{lem:qgraphs_linegraphs} uses known transformations between linear combinations
of homomorphisms, subgraphs and induced subgraphs~(see for instance\ Chapter~5.2.3 in~\cite{Lovasz12}
and~\cite[Section~3]{CurticapeanDM17}), as well as Whitney's Isomorphism Theorem:
\begin{theorem}[\cite{Whitney92}]\label{thm:whitney}
    Let $H$ be a connected line graph that is not isomorphic to the triangle.
    Then the graph $F$ such that $L(F)=H$ is uniquely defined up to isolated vertices.
    More precisely, every graph $F'$ that satisfies $L(F')=H$ and that does not contain isolated
    vertices is isomorphic to $F$.\lipicsEnd
\end{theorem}

\begin{figure}
    \centering
    \begin{tikzpicture}
        \pgfmathsetmacro{\d}{3}

        \pic[kneser style=13,\icol] at (-6.5,0) {kneser=\d/1};
        \pic[kneser style=13,vertex fill color=\icol] at (-\d,0) {kneser=\d/1};

        \node at (-\d,-\d/2.5) {$L(K_3) = L(K_{1,3})$};

        \draw[line width=1pt,-{Latex[length=5mm, width=2mm]}]
        (-.5,0) --node[above,midway] {\small Line graph} (-\d/1.5,0);
        \draw[line width=1pt,-{Latex[length=5mm, width=2mm]}]
        (-5.5,0) --node[above,midway] {\small Line graph} (-4,0);

        \node[vertex] (a) at (0,0) {};
        \foreach\i in {1,...,\d}{
            \node[vertex] (b\d) at (1,-.5*.6+-.6*\d/2+\i*.6) {};
            \draw[thick,\icol] (b\d) -- (a);
        }
        \node at (.5,-\d/2.5) {$K_{1,3}$};
        \node at (-6.5,-\d/2.5) {$K_3$};
    \end{tikzpicture}
    \caption{The claw $K_{1,3}$ and the triangle $K_3$ have the same line graph $K_3$.}\label{fig:cltr}
\end{figure}
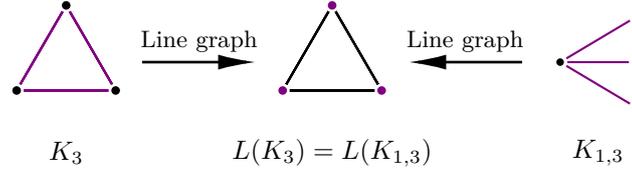

Note that both the triangle and the claw have the triangle as line graph, consider \cref{fig:cltr}.
The previous theorem states that the triangle is the only line graph whose primal graph is not
uniquely defined.
Further, the Isomorphism Theorem (\cref{thm:whitney}) allows us to define the following
function; recall that $\mathcal{G}_{P}$ is the set of bipartite graphs.
\[ L^{-1}: \symking \rightarrow \mathcal{G}_{P},\text{ such that $L^{-1}(L(G)) := G$
for every bipartite graph $G$.} \]
Note that the function $L^{-1}$ is well-defined by \cref{thm:whitney} and the fact that the
triangle is not bipartite.
In particular, we have that $L^{-1}(K_3)= K_{1,3}$, where $K_3$ is the triangle and $K_{1,3}$
is the claw, that is, the complete bipartite graph with one vertex on the left side and three
vertices on the right side (Again, consult \cref{fig:cltr} for a visualization.)
Similarly, it is well defined to write $L(F)$ for a Kőnig graph which is the line graph of the
(uniquely determined) bipartite graph $F$ without isolated vertices.

Further, again by Whitney's Isomorphism Theorem, we obtain the following lemma.
\begin{lemma}
    Let $H$ be a line graph and let $L(G)$ be a Kőnig graph. Then we have that
    \begin{equation}
        \#\indsubs{H}{L(G)} = \begin{cases} \#\subs{F}{G} ~~&\text{if } H=L(F)\in \symking \\ 0& \text{otherwise} \end{cases}. \label{eq:line_subs_indsubs}
    \end{equation}
\end{lemma}
\begin{proof}
    Assume first that the graph $H=L(F)$ is a Kőnig graph.
    By \cref{thm:whitney} and the fact that triangles are not bipartite,
    we have that $G$ is the unique bipartite graph whose line graph is isomorphic to $L(G)$.
    Let $S\subseteq V(L(G))= E(G)$ such that the induced subgraph $L(G)[S]$ is isomorphic to
    $L(F)$ and let $G'$ be the subgraph of $G$ with vertices
    \[ V(G'):= \{v \in V(G)\mid \exists e\in S: v \in e\},\]
    and edges $E(G'):=S$.

    By construction, the line graph of $L(G')$ is isomorphic to $L(F)$.
    Hence, we obtain that the graphs $F$ and $G'$ are isomorphic by \cref{thm:whitney}
    and the fact that none of the graphs $F$ and $G'$ is the triangle;
    note that a bipartite graph cannot contain a subgraph isomorphic to a triangle.

    Now let $G'$ be a subgraph of $G$ that is isomorphic to $F$.
    As the graph $F$ does not contain isolated vertices, we have that $G'$ is determined by its
    set of edges. We set $S=E(G')$ and consider the induced subgraph $L(G)[S]$.
    By construction, we have that the graph $L(G)[S]$ is isomorphic to $L(G')$ which in turn
    is isomorphic to $L(F)$ (as the graphs $G'$ and $F$ are isomorphic).
    This shows correctness for $H=L(F)\in \symking$.

    If the graph $H$ is not a Kőnig graph, then $H$ is not a triangle.
    Thus, we have that $H=L(F')$ for some non-bipartite graph $F'$ which is uniquely determined
    (up to isolated vertices) by \cref{thm:whitney}.
    In this case, the same argument as before shows that any induced subgraph of $L(G)$
    that is isomorphic to $L(F')$ yields a subgraph of $G$ that is isomorphic to $F'$.
    As the graph $G$ is bipartite and the graph $F'$ is not, such an induced subgraph cannot
    exist; note that bipartite graphs are closed under taking subgraphs.
\end{proof}

The last ingredient for the proof of \cref{lem:qgraphs_linegraphs} is the following well-known
identity which relates (strong) embeddings and (induced) subgraphs.
\begin{fact}\label{fac:auts_trans}
    For all graphs $H$ and $G$ we have that
    \begin{align}
        \#\embs{H}{G} = \#\auts{H}\cdot \#\subs{H}{G}\text{ and}\label{eq:embs_subs}\\
        \#\strembs{H}{G} = \#\auts{H}\cdot \#\indsubs{H}{G}.\label{eq:strembs_indsubs}
    \end{align}\lipicsEnd
\end{fact}

Finally, we are ready to prove \cref{lem:qgraphs_linegraphs}.
\begin{proof}[Proof of \cref{lem:qgraphs_linegraphs}]
    We rely on a stepwise transformation of linear combinations of homomorphisms,
    embeddings, and strong embeddings as given by \lovasz~\cite[Chapter~5.2.3]{Lovasz12}
    and as used by Curticapean, Dell and Marx~\cite[Section~3]{CurticapeanDM17}.

    For the formal statement, we need to introduce some further notation.
    Given a graph $H$, we write $\mathsf{Part}(H)$ for the set of all partitions $\rho$ of
    the vertex set $V(H)$ such that the quotient $H/\rho$ is a spasm, that is,
    $H/\rho$ does not contain self-loops.
    Further, we write $H' \supseteq H$ if the graph $H'$ can be obtained from $H$ by adding
    edges. Now, let us state the transformations:
    For all graphs $H$ and $G$, we have that
    \begin{align}
        \#\homs{H}{G} &= \sum_{\rho \in \mathsf{Part}(H)} \#\embs{H/\rho}{G},\text{ and }\label{eq:hom_to_emb}\\
        \#\embs{H}{G} &= \sum_{\rho \in \mathsf{Part}(H)} \mu(\emptyset,\rho)\cdot \#\homs{H/\rho}{G},\label{eq:emb_to_hom}
    \end{align}
    where $\mu$ is the so-called Möbius function over the partition
    lattice.\footnote{As we do not need the formal definition of the Möbius function here,
    we refer the interested reader to for instance\ \cite{Stanley11} instead.}
    Further, we have that
    \begin{align}
        \#\embs{H}{G} &= \sum_{H' \supseteq H} \#\strembs{H'}{G},\text{ and }\label{eq:emb_to_stremb}\\
        \#\strembs{H}{G} &= \sum_{H' \supseteq H} (-1)^{|E(H')|-|E(H)|} \cdot \#\embs{H'}{G}.\label{eq:stremb_to_emb}
    \end{align}
    Now let us start with the construction of the quantum graph $Q[H]$ and the proof of
    Equation~\eqref{eq:line_graph_a}.  We have that
    \begin{align*}
        \#\homs{H}{L(G)}\stackrel{\eqref{eq:hom_to_emb}}{=}&  \sum_{\rho \in \mathsf{Part}(H)} \#\embs{H/\rho}{L(G)}\,,\\
        \#\embs{H/\rho}{L(G)} \stackrel{\eqref{eq:emb_to_stremb}}{=}& \sum_{H' \supseteq H/\rho} \#\strembs{H'}{L(G)} \,,\text { and } \\
         \#\strembs{H'}{L(G)} \stackrel{\eqref{eq:strembs_indsubs}}{=}&  \#\auts{H'} \cdot \#\indsubs{H'}{L(G)}\,.
    \end{align*}
    Combining \eqref{eq:line_subs_indsubs} with the fact that line graphs are closed under taking
    induced subgraphs~\cite{Beineke70}, we obtain that \[
        \#\indsubs{H'}{L(G)} = 0,
    \] whenever $H'$ is not a Kőnig graph.
    Consequently,
    \begin{align*}
        \#\embs{H/\rho}{L(G)}=\sum_{\substack{H' \supseteq H/\rho\\ H'=L(F)\in \symking }} \#\auts{H'} \cdot \#\indsubs{H}{L(G)}.
    \end{align*}
    For a graph $H'= L(F) \in \symking$ we have that
    \begin{align*}
        \#\indsubs{H'}{L(G)} &\stackrel{\eqref{eq:line_subs_indsubs}}{=}   \#\subs{F}{G} \\
                             &\stackrel{\eqref{eq:embs_subs}}{=}   {\#\auts{F}}^{-1} \cdot \#\embs{F}{G}\\
                             &\stackrel{\eqref{eq:emb_to_hom}}{=}   {\#\auts{F}}^{-1} \cdot \sum_{\rho \in \mathsf{Part}(F)} \mu(\emptyset,\rho) \cdot \#\homs{F/\rho}{G}.
    \end{align*}
    We can thus successively apply the previous transformations and obtain:
\[
\#\homs{H}{L(G)} = \sum_{\rho \in \mathsf{Part}(H)} \sum_{\substack{H' \supseteq H/\rho\\ H'=L(F)\in \symking }}  \sum_{\delta \in \mathsf{Part}(F)} \frac{\#\auts{H'} \cdot \mu(\emptyset,\delta)}{\#\auts{F}} \cdot \#\homs{F/\delta}{G} \,.
\]
The desired quantum graph $Q[H]$ is obtained by collecting for isomorphic terms, that is,
\[Q[H] := \sum_{F'} \lambda_{F'} \cdot F' \,,\]
where
\[ \lambda_{F'} := \sum_{\rho \in \mathsf{Part}(H)} \sum_{\substack{H' \supseteq H/\rho\\ H'=L(F)\in \symking }}  \sum_{\substack{\delta \in \mathsf{Part}(F)\\ F/\delta = F'}} \frac{\#\auts{H'} \cdot \mu(\emptyset,\delta)}{\#\auts{F}} \,.\]
\end{proof}

Using \cref{lem:qgraphs_linegraphs}, we obtain a proof for \cref{thm:implicit_line_dicho}
as follows.
{
\renewcommand{\thetheorem}{\ref{thm:implicit_line_dicho}}
\begin{theorem}[repeated]
    Let $\mathcal{H}$ denote a recursively enumerable class of graphs.
    Then the problem $\#\homsprob{\mathcal{H}}{\symking}$ is either fixed-parameter tractable
    or $\#\W{1}$-hard under parameterized Turing-reductions.
\end{theorem}
\addtocounter{lemma}{-1}
}
\begin{proof}
    Let $P = P_1$ denote any path of length $1$, and let $\mathcal{G}_P$ denote the
    class of all $P$-colorable graph; that is, $\mathcal{G}_P$ is the class of all
    bipartite graphs. Now, consider the following class of graphs\[
        \hat{\mathcal{H}} := \bigcup_{H\in\mathcal{H}}\mathsf{supp}(Q[H]) \cap \mathcal{G}_{P},
    \]  We show the following reductions\[
        \#\homsprob{\mathcal{H}}{\symking}\interred \#\homsprob{\hat{\mathcal{H}}}{\mathcal{G}_{P}}.\]
    Note that this imples \cref{thm:implicit_line_dicho} by the classification of
    $F$-colorable graphs (\cref{thm:genhomdich}).

    For the direction $\#\homsprob{\mathcal{H}}{\symking} \fptred
    \#\homsprob{\hat{\mathcal{H}}}{\mathcal{G}_{P}}$, we assume that a graph $H\in \mathcal{H}$
    and a Kőnig graph $L(G)$ are given.
    By \cref{lem:main_line_graphs}, we can compute (in time depending only on $H$)
    the quantum graph $Q[H]$ that satisfies
    \begin{equation}\label{eq:line_dicho_eq}
        \#\homs{H}{L(G)} = \#\homs{Q[H]}{G} = \sum_{F \in \mathsf{supp}(Q[H])} \lambda_F \cdot \#\homs{F}{G},
    \end{equation}
    where the $\lambda_F$ are the coefficients of $Q[H]$.
    Now, as $L(G)$ is a Kőnig graph, we have that $G$ is bipartite.
    Hence, by \cref{lem:teclem1}, there is no homomorphism from $F$ to $G$ whenever $F$ is not
    bipartite, that is we have that $\#\homs{F}{G} = 0$.
    Note that we can verify whether a graph $F \in \mathsf{supp}(Q[H])$ is bipartite in time
    depending only on $|V(H)|$.
    All further terms $\#\homs{F}{G}$ with $F \in \mathcal{G}_{P}$ can be obtained by querying
    the oracle for the problem $\#\homsprob{\hat{\mathcal{H}}}{\mathcal{G}_{P}}$.
    Finally, we compute and the linear combination given by~\eqref{eq:line_dicho_eq}.
    This completes the first reduction.

    For the other direction, $\#\homsprob{\hat{\mathcal{H}}}{\mathcal{G}_{P}} \fptred
    \#\homsprob{\mathcal{H}}{\symking}$, we assume that graphs $F \in \hat{\mathcal{H}}$
    and $G \in \mathcal{G}_{P_2}$ are given.
    By definition of the class $\hat{\mathcal{H}}$, we have that the graph $F$ is a bipartite
    constituent of the quantum graph $Q[H]$ for some $H \in \mathcal{H}$.
    As $\mathcal{H}$ is recursively enumerable and the mapping $H \mapsto Q[H]$ is computable,
    we can compute in time depending only on $|V(F)|$ the quantum graph $Q[H]$.

    By \cref{fac:f_col_tensor}, we have that the tensor product $G \times A$ is bipartite
    for every graph $A$ as the graph $G$ is bipartite. Therefore,
    the graph $L(G \times A)$ is a Kőnig graph for every (not necessarily bipartite) graph $A$.
    Hence, we can query the oracle for the problem $\#\homsprob{\mathcal{H}}{\symking}$
    to compute for every graph $A$ whose size depends only on $|V(F)|$ the following
    values:
    \begin{align*}
        \#\homs{H}{L(G \times A)} =& \#\homs{Q[H]}{G \times A}\\
                                  =& \sum_{F' \in \mathsf{supp}(Q[H])} \lambda_{F'} \cdot
                                  \#\homs{F'}{G\times A}\\
                                  =& \sum_{F' \in \mathsf{supp}(Q[H])} \lambda_{F'} \cdot
                                  \#\homs{F'}{G} \cdot \#\homs{F'}{A}.
    \end{align*}
    Now, (the proof of) \cref{lem:f_col_monotonicity} shows that the induced system of linear
    equations is solvable for the proper choices of $A$.
    In particular, the size of those choices depends only on $|V(H)|$, which itself
    depends only
    on~$|V(F)|$. As the graph $F$ is a constituent of the quantum graph $Q[H]$, and thus
    the corresponding coefficient $\lambda_F$ is non-zero, we can compute and output
    the number $\#\homs{F}{G}$ in time depending only on $|V(H)|$, which itself depends
    only on~$|V(F)|$.
    This completes the second reduction, and hence the proof.
\end{proof}

\clearpage
\bibliography{conference}

\clearpage

\appendix
\clearpage

\section{On Hardness of \texorpdfstring{$\#\homsprob{\mathcal{H}}{\top}$}{|Hom(H->T)|}}\label{sec:hom_hard_strategy}
In this section, we take a closer look at the proof of the following complexity classification
which is due to Dalmau and Jonsson.
\begin{theorem}[\cite{DalmauJ04}]\label{thm:app_class}
    Let $\mathcal{H}$ be a recursively enumerable class of graphs.
    \begin{enumerate}[~1~]
        \item If the treewidth of $\mathcal{H}$ is bounded then
            $\#\homsprob{\mathcal{H}}{\top}$ is solvable in polynomial time.
        \item Otherwise, $\#\homsprob{\mathcal{H}}{\top}$ is $\#\W{1}$-hard
            under parameterized Turing-reductions.\lipicsEnd
    \end{enumerate}
\end{theorem}

In particular, we are interested in the proof of Statement~$(2)$ of \cref{thm:app_class},
that is $\#\W{1}$-hardness for the problem $\#\homsprob{\mathcal{H}}{\top}$.
The strategy of the proof is a line of reasoning based on the
Excluded Grid Theorem\footnote{Recall that the Excluded Grid Theorem states that every class
$\mathcal{H}$ of unbounded treewidth contains arbitrarily large grid
minors~\cite{RobertsonS86-ExGrid}.}, which is, by now, well-established
(see for instance\ \cite{DalmauJ04,GroheSS01,Curticapean15,DellRW19icalp}).
Our goal in this section is to show the following consequences of the known proofs of
Statement~$(2)$ of \cref{thm:app_class}; recall that a core is a graph without a homomorphism
from itself to any of its proper subgraphs.
\begin{lemma}\label{lem:help_hardness}
    Let $\mathcal{H}$ be a recursively enumerable class of graphs of unbounded treewidth.
    \begin{enumerate}[~1~]
        \item There is a parameterized Turing-reduction from the problem $\#\clique$ to
            the problem $\#\homsprob{\mathcal{H}}{\top}$ such that every oracle query
            $(H,G)$ satisfies that the graph $G$ is $H$-colorable.
        \item If, additionally, the class $\mathcal{H}$ contains only connected cores, then
            there is a parameterized weakly parsimonious reduction from the problem
            $\#\clique$ to the problem $\#\homsprob{\mathcal{H}}{\top}$ such that every pair
            $(H,G)$ in the image of the reduction satisfies that the graph $G$ is connected
            and $H$-colorable.\lipicsEnd
    \end{enumerate}
\end{lemma}

To accommodate readers unfamiliar with the proof of \cref{thm:app_class} on the one hand,
but to refrain from including the proof in full detail on the other hand, we provide
an outline only. For technical reasons, we start by proving that the instances of the problem
$\#\clique$ can be assumed to be connected graphs.
\begin{lemma}\label{lem:connclique}
    Let $\#\textsc{ConnClique}$ be the problem of, given a \emph{connected} graph $G$ and a
    positive integer $k$, computing the number of cliques of size $k$ in $G$. Then we have that
    \[\#\clique \parsired \#\textsc{ConnClique}.\]
\end{lemma}
\begin{proof}
    Let $(G,k)$ be an instance of $\#\clique$.
If the number $k$ is $1$, then the number of $k$-cliques in $G$ is just the number
    of vertices $|V(G)|$ of $G$. Hence, the reduction can output $(P_{|V(G)|-1},2)$, where
    $P_i$ is the path with $i+1$ edges.

    If the number $k$ is $2$, then the number of $k$-cliques in $G$ is just the number
    of edges $|E(G)|$ of $G$. Hence, the reduction can output $(P_{|E(G)|-1},2)$.

    Otherwise, let $C_1\dots,C_n$ be the connected components of $G$.
    For each $i\in{1,\dots,n-1}$, we add an edge between an arbitrary vertex in component $C_i$
    and an arbitrary vertex in component $C_{i+1}$.
    As the number $k$ is at least $3$, this operation does not change the number of $k$-cliques
    and thus the reduction can output the modified connected graph and $k$.
\end{proof}

\begin{proof}[{Outline of the proofs of~(2) of \cref{thm:app_class} and \cref{lem:help_hardness}}]
    Let $\mathcal{H}$ denote a recursively enumerable class of graphs of unbounded treewidth.
    The goal is to show that the problem $\#\homsprob{\mathcal{H}}{\top}$ is $\#\W{1}$-hard.

    The first step is the reduction from $\#\clique$ to $\#\cphomsprob{\boxplus}{\top}$, where
    $\boxplus$ is the set of all $k \times k$ square grids $\boxplus_k$ for $k\in \mathbb{N}$.
    Intuitively, given an instance $(G,k)$ of $\#\clique$, we construct a $\boxplus_k$-colored
    graph~$G'$ as follows: For each $i\in\{1,\dots,k\}$, we add the set
    \[V_{i,i}:=\{(v,v)\mid v \in V(G)\}\] to the vertices of $G'$.
    For every $i,j\in\{1,\dots,k\}$ with $i \neq j$, we add the set
    \[V_{i,j} := \{(u,v)\mid  \{u,v\} \in E(G)\}\] to the vertices of $G'$.
    Finally, we add an edge between two vertices $(v,u) \in V_{i,j}$ and $(v',u') \in V_{i',j'}$
    if $v = v'$ and $i = i'$, or if $u=u'$ and $j=j'$.
    It can the be shown that the number of $k$-cliques in $G$ equals the number
    \[\#\cphoms{\boxplus_k}{G'}\cdot (k!)^{-1},\]
    where the factor $(k!)^{-1}$ stems from the fact that the vertices of a $k$-clique are not
    ordered.
    Further, the resulting graph $G'$ is connected if $G$ is connected and
    $G'$ is $\boxplus_k$-colorable given by the homomorphism $h$ that maps every vertex in
    $V_{i,j}$ to the grid vertex $(i,j)$.

    The second step relies on the Excluded Grid Theorem \cite{RobertsonS86-ExGrid}:
    If the class $\mathcal{H}$ has unbounded treewidth, then for any number $k$, there is a
    graph $H_k \in \mathcal{H}$ that has the grid $\boxplus_k$ as a minor.
    Using this property of the class $\mathcal{H}$, it can be shown that
    \[\#\cphomsprob{\boxplus}{\top} \parsired \#\cphomsprob{\mathcal{H}}{\top}.\]
    A very clear presentation of this reduction is given by Curticapean~\cite[Lemma~5.8]{Curticapean15}.\footnote{Note that Curticapean uses the problem
    $\#\textsc{PartitionedSub}(\mathcal{H})$ which, however, is equivalent to the problem $\#\cphomsprob{\mathcal{H}}{\top}$~\cite[Definition~5.2 and Remark~5.3]{Curticapean15}.}
    In particular, given graphs $\boxplus_k$ and $G$, the reduction outputs a pair $(H_k,\hat{G})$
    such that the graph $\hat{G}$ is $H_k$-colorable and $\hat{G}$ is connected if both $G$ and
    $H_k$ are connected. Further, note that the graph $H_k$ can be found as
    $\mathcal{H}$ is recursively enumerable.

    We are now able to prove the second item of \cref{lem:help_hardness}:
    Using \cref{lem:connclique} and the properties of the previous reductions, we can, on
    input $G$ and $k$, compute in time $f(k)\cdot |V(G)|^{O(1)}$ (for some computable function $f$),
    a pair $(H_k,\hat{G})$ of graphs such that
    \begin{enumerate}[(a)]
        \item the graph $\hat{G}$ is connected,
        \item the graph $\hat{G}$ is $H_k$-colorable by some coloring $c$, and
        \item the number of $k$-cliques in $G$ is precisely $\#\cphoms{H_k}{\hat{G}}\cdot (k!)^{-1}$.
    \end{enumerate}
    Now recall that the condition of the second item of \cref{lem:help_hardness} states that
    the graph $H_k$ is a core.
    Let $h$ be a (not necessarily color-prescribed) homomorphism from $H_k$ to $\hat{G}$.
    Then, the composition of $h$ and the coloring $c$ is an endomorphism of $H_k$.
    As the graph $H_k$ is a core, it has no homomorphism to a proper subgraph. Hence,
    the homomorphism $h \circ c$ is an automorphism and, in particular, $h$ is color-prescribed
    if and only if $h\circ c$ is the identity.
    It is then straightforward\footnote{A reader familiar with group actions will find a very easy
        proof based on the observation that the automorphism group of $H_k$ acts on the set
    $\homs{H_k}{\hat{G}}$.} to show that
    \[\#\cphoms{H_k}{\hat{G}} = \#\homs{H_k}{\hat{G}} \cdot \#\mathsf{Aut}(H_k)^{-1}.\]
    This concludes the proof of the second item of \cref{lem:help_hardness}.

    For the first item, we cannot assume that the graph $\hat{G}$ is connected, as the graph $H_k$
    might be disconnected.
    However, we still obtain an algorithm that, on input $G$ and $k$, computes in time
    $f(k)\cdot |V(G)|^{O(1)}$ (for some computable function $f$),
    a pair $(H_k,\hat{G})$ of graphs such that
    \begin{enumerate}[(i)]
        \item the graph $\hat{G}$ is $H_k$-colorable by some coloring $c$, and
        \item the number of $k$-cliques in $G$ is precisely $\#\cphoms{H_k}{\hat{G}}\cdot (k!)^{-1}$.
    \end{enumerate}
    Now, using the principle of inclusion and exclusion, it is possible to compute the number $N$
    of homomorphisms $h$ from $H_k$ to $\hat{G}$ such that $h \circ c$ is an automorphism in
    time $O(2^k)\cdot|V(\hat{G})|^{O(1)}$ if an oracle to $\#\cphoms{H_k}{\star}$ is provided.
    More precisely, we have that
    \[N = \sum_{J\subseteq V(H_k)} (-1)^{|J|} \cdot \#\homs{H_k}{\hat{G}\setminus J},\]
    where $\hat{G}\setminus J$ is the graph obtained from $\hat{G}$ by deleting all vertices $v$ for
    which $c(v)\in J$. Having obtained $N$, for the same reasons as in the previous case,
    the Turing-reduction can output
    \[N \cdot \#\mathsf{Aut}(H_k)^{-1}.\]
    The first item of \cref{lem:help_hardness} now holds as the graph $\hat{G}$ is $H_k$-colorable
    and hence every subgraph $\hat{G}\setminus J$ is $H_k$-colorable as well.
\end{proof}

Note that the proofs of \cref{lem:connclique} and the second item of \cref{lem:help_hardness}
immediately show the following consequence for the decision version and parameterized
many-one reductions~\cite[Definition~2.1]{FlumG06}:
\begin{corollary}\label{cor:decision_hard}
    Let $\mathcal{H}$ be a recursively enumerable class of connected cores of unbounded treewidth.
    Then there is a parameterized many-one reduction from $\clique$ to
    $\homsprob{\mathcal{H}}{\top}$ such that every pair $(H,G)$ in the image of the reduction
    satisfies that $G$ is connected and $H$-colorable.\lipicsEnd
\end{corollary}

\section{Proofs of Section~\ref{sec:fcol}}\label{sec:app_fcol}
\subsection{Proof of Theorem~\ref{thm:genhomdich}}\label{sec:proof_f_homsdicho}

We prove \cref{thm:genhomdich} in two steps. First, in \cref{lem:pol},
we show a polynomial-time algorithm for graph classes $\mathcal{H}_{easy}$ that do not
contain $F$-colorable graphs of arbitrarily large treewidth.
After that, we prove $\#W[1]$-hardness for all other graph classes.

\paragraph*{Polynomial-Time Algorithm for the Tractable Cases}

Let $\mathcal{H}_{easy}$ denote any graph class such that for any graph $H \in \mathcal{H}_{easy}$,
either $H$ has a treewidth of at most $c$, or $H$ is not $F$-colorable;
where $c = c(\mathcal{H}_{easy})$ is a constant depending only  on $\mathcal{H}_{easy}$.
We obtain a polynomial-time algorithm for the (PPC) problem
$\#\homs{\mathcal{H}_{easy}}{\mathcal{G}_{F}}$ as follows. Given graphs $H \in \mathcal{H}_{easy}$
and $G \in \mathcal{G}_{F}$, we check, using Bodlaender's Algorithm~\cite{Bodlaender96},
whether $H$ has a treewidth $\tw(H)$ of at most $c(\mathcal{H}_{easy})$.
Next, if $\tw(H) \le c(\mathcal{H}_{easy})$, we use the standard dynamic programming algorithm
due to {D{\'{\i}}az} et.\ al~\cite{DiazST02} to compute $\#\homs{H}{G}$. Otherwise, that is if
$\tw(H) > c(\mathcal{H}_{easy})$, we output 0, as $H$ is not $F$-colorable by definition of
$\mathcal{H}_{easy}$. This last step is justified by the following observation.

\begin{observation}\label{lem:teclem1}
    Let graphs $F, G,$ and $H$ be given. If there is no homomorphism from $H$ to $F$,
    but a homomorphism from $G$ to $F$, then there is no homomorphism from $H$ to $G$.
\end{observation}
\begin{proof}
    Choose any homomorphism $g$ from $G$ to $F$ and suppose there was a homomorphism~$h$ from $H$
    to $G$. As the concatenation of two homomorphisms is again a homomorphism,
    in particular $f \circ g : H \to F$ is again a homomorphism, which is a contradiction
    to the assumption that there is no homomorphism from $H$ to $F$.
\end{proof}

In total, we obtain the following algorithm.

\begin{lemma}\label{lem:pol}
    For any graph classes $\mathcal{H}_{easy}$ and $\mathcal{G}_{F}$, we can solve
    the (PPC) problem $\#\homs{\mathcal{H}_{easy}}{\mathcal{G}_{F}}$ in polynomial time.
\end{lemma}
\begin{proof}
    The correctness follows directly from the definition of $\mathcal{H}_{easy}$ and
    \cref{lem:teclem1}.

    For the running time, set $c = c(\mathcal{H}_{easy})$, $k = |V(H)|$, and $n := |V(G)|$
    Checking whether the given graph has a treewidth of at most $c$ takes time
    $c^{O(c^3)} \cdot k$ using Bodlaender's algorithm\cite{Bodlaender96}.
    Next, computing the number of homomorphisms from a graph with treewidth at most
    $c$ takes time ${\rm poly}(k, c) \cdot n^{c + O(1)}$ due to D{\'{\i}}az et.\ al~\cite{DiazST02}.
    Hence in total, our algorithm has a running time of ${\rm poly}(k, c) \cdot n^{c + O(1)}$,
    which is polynomial, thus completing the proof.
\end{proof}

\paragraph*{$\#W[1]$-Hardness for the Intractable Cases}

It remains to demonstrate $\#\W{1}$-hardness of $\#\homs{H}{\mathcal{G}_F}$ whenever the treewidth
of the intersection $\mathcal{H} \cap \mathcal{G}_F$ is unbounded. However, as we have seen in
\cref{lem:help_hardness}, the existing hardness proof already shows the desired stronger result.

\begin{observation}\label{cor:f_col_hard}
    Let $\mathcal{H}$ be a recursively enumerable class of graphs such that the treewidth of
    $\mathcal{H}\cap \mathcal{G}_F$ is unbounded. Then we have that
    \[\#\clique \fptred \#\homsprob{\mathcal{H}}{\mathcal{G}_F}.\]
\end{observation}
\begin{proof}
    We use \cref{lem:help_hardness} (1) for the class $\mathcal{H}\cap \mathcal{G}_F$:
    As every oracle query $(H,G)$ satisfies that $G$ is $H$-colorable by some coloring $c$
    and every graph $H \in \mathcal{H}\cap \mathcal{G}_F$ is $F$-colorable by some coloring $c'$,
    we obtain that the composition $c\circ c'$ is an $F$-coloring of $G$.
\end{proof}

{
\renewcommand{\thetheorem}{\ref{thm:genhomdich}}
\begin{theorem}[repeated]
    Let $F$ be a graph, and let $\mathcal{H}$ be a recursively enumerable class of graphs.
    \begin{enumerate}[~1~]
    \item If the treewidth of $\mathcal{H}\cap\mathcal{G}_F$ is bounded then the PPC problem
        $\#\homsprob{\mathcal{H}}{\mathcal{G}_F}$ is polynomial-time solvable.
    \item Otherwise, the problem $\#\homsprob{\mathcal{H}}{\mathcal{G}_F}$ is $\#\W{1}$-hard.
    \end{enumerate}
\end{theorem}
\addtocounter{lemma}{-1}
}
\begin{proof}
    Holds by \cref{lem:pol} and \cref{cor:f_col_hard}.
\end{proof}

\subsection{Proof of Theorem \ref{thm:subs_f_colored}}\label{sec:app_fcol_sub}

We start with the proof of the following lemma.
\begin{lemma}\label{lem:f_col_monotonicity}
    Let $F$ be a fixed graph and let $Q$ be a quantum graph such that
    $\mathsf{supp}(Q)\subseteq \mathcal{G}_F$.
    There is a deterministic algorithm $\mathbb{A}$ that is given oracle access to
    the problem $\#\homsprob{Q}{\star}$ and that, on input an $F$-colorable graph $G$,
    computes the number $\#\homs{H}{G}$ for every constituent $H$ of $Q$.
    Further, there are computable functions $f$ and $s$ such that the running time
    of the algorithm $\mathbb{A}$ is bounded by $f(|Q|)\cdot |V(G)|^{O(1)}$ and every graph~$G'$
    for which the oracle is queried is $F$-colorable and has at most $s(|Q|)\cdot |V(G)|$ vertices.
\end{lemma}
\begin{proof}
    We follow the lines of the proof of Lemma~3.6 in~\cite{CurticapeanDM17}:
    Given an $F$-colorable graph $G$, we wish to query the oracle for $(Q,G\times H')$ for
    certain graphs $H'$. Recall that the tensor product satisfies
    \[\#\homs{A}{B\times C} = \#\homs{A}{B}\cdot \#\homs{A}{C},\]
    for all graphs $A,B$ and $C$. Curticapean, Dell and Marx~\cite{CurticapeanDM17} discovered
    that a deep result of \lovasz~implies the existence of graphs $H_1,\dots,H_\ell$ such that the
    following system of linear equations has a unique solution.
    \[\#\homs{Q}{G\times H_i} = \#\sum_H \lambda_H \cdot \#\homs{H}{G \times H_i} = \#\sum_H c_H \cdot   \#\homs{H}{H_i},\]
    where $c_H := \lambda_H \cdot \#\homs{H}{G}$.
    In particular, the graphs $H_1,\dots,H_\ell$ can be computed in time depending only on $H$.
    Consequently the number $\#\homs{H}{G}$ can be computed whenever $\lambda_H \neq 0$ by
    standard Gaussian elimination.

    Now, the only catch is the fact that we are allowed to query the oracle only for $F$-colorable
    graphs. However, by \cref{fac:f_col_tensor}, the tensor product of an $F$-colorable graph with
    another (not necessarily $F$-colorable) graph is always $F$-colorable. Consequently, the
    original proof of Curticapean, Dell and Marx~\cite{CurticapeanDM17} transfers without
    modification to the $F$-colored setting.
\end{proof}

{
\renewcommand{\thetheorem}{\ref{thm:subs_f_colored}}
\begin{theorem}[repeated]
    Let $F$ be a fixed graph and let $\mathcal{H}$ be a recursively enumerable class of graphs.
    \begin{enumerate}[~1~]
        \item If the matching number of $\mathcal{H}\cap \mathcal{G}_F$ is bounded then
            the problem $\#\subsprob{\mathcal{H}}{\mathcal{G}_F}$ is polynomial-time solvable.
        \item Otherwise, the problem $\#\subsprob{\mathcal{H}}{\mathcal{G}_F}$ is $\#\W{1}$-hard.
    \end{enumerate}
\end{theorem}
\addtocounter{lemma}{-1}
}
\begin{proof}
    Proving~(1) is easy: Let $c$ be the constant upper bound on the matching number of graphs in
    $\mathcal{H}\cap \mathcal{G}_F$. Now, given $H\in \mathcal{H}$ and $G \in \mathcal{G}_F$,
    we first compute the matching number of $H$ in polynomial time by the
    Blossom-Algorithm~\cite{Edmonds65}. If the result is greater than $c$, the promise tells us
    that $H \notin \mathcal{G}_F$, in which case we can output $0$ as $H$ would be $F$-colorable
    if it was isomorphic to a subgraph of $G\in \mathcal{G}_F$. Otherwise, we use the algorithm
    given by the first item of \cref{thm:subsdicho}.

    For $\#\W{1}$-hardness, we construct a reduction from the problem
    \[\#\homsprob{\mathsf{spasms}(\mathcal{H})\cap \mathcal{G}_F}{\mathcal{G}_F},\]
    where $\mathsf{spasms}(\mathcal{H})$ is the set of all spasms of graphs in $\mathcal{H}$.
    We start with the following observation.

    \begin{claim}
        The class $\mathsf{spasms}(\mathcal{H}) \cap \mathcal{G}_F$ has unbounded treewidth if
        the class $\mathcal{H}\cap \mathcal{G}_F$ has unbounded matching number.
    \end{claim}
    \begin{claimproof}
        Let $b \in \mathbb{N}$ be a positive integer.
        We show that there is a graph of treewidth at least $b$ in the class
        $\mathsf{spasms}(\mathcal{H}) \cap \mathcal{G}_F$.
        By assumption, there is a graph $H \in \mathcal{H}\cap \mathcal{G}_F$ such that
        $H$ contains a matching $M$ of size at least $|E(F)| \cdot b^2$;
        recall that $|E(F)|$ is a constant as the graph $F$ is fixed.

        Now let $c$ be an $F$-coloring of $H$.
        For every edge $e=\{u,v\} \in M$, we let $c(e)\in E(F)$ be the image of $e$ under
        the coloring.
        As $|M|\geq |E(F)| \cdot b^2$, we have that there is an edge
        $\hat{e}=\{\hat{u},\hat{v}\} \in E(F)$ such that $c(e)=\hat{e}$ for at least $b^2$
        many edges of $M$.
        More precisely, there are edges $\{u_1,v_1\}, \dots,\{u_{b^2},v_{b^2}\} \in M$ such that
        $c(u_i)=\hat{u}$ and $c(v_i)=\hat{v}$ for all $i \in \{1,\dots,b^2\}$.
        As $c$ is a coloring (and thus a homomorphism), we have that the sets
        $U:=\{u_1,\dots,u_{b^2}\}$ and $V:=\{v_1,\dots,v_{b^2}\}$ are independent.
        Consequently, we can contract vertices in $U$ and $V$ such that the matching $M$ becomes
        a complete bipartite graph $K_{b,b}$ with $b$ vertices on each side.

        Let $\rho$ be the induced partition on $V(H)$. Then the quotient graph $H/\rho$ contains
        as subgraph $K_{b,b}$, and $H/\rho$ is still $F$-colorable as we identified only vertices
        that are contained in the same pre-image of the $F$-coloring $c$.
        Further, $H/\rho$ is a spasm as we did not create self-loops.
        Thus, we have that $H/\rho \in \mathsf{spasms}(\mathcal{H}) \cap \mathcal{G}_F$.
        As the treewidth of $K_{b,b}$ is $b$ and the treewidth of a graph cannot increase by
        taking subgraphs, we obtain that $H/\rho$ has treewidth at least $b$.
    \end{claimproof}
    Consequently, a reduction from
    $\#\homsprob{\mathsf{spasms}(\mathcal{H})\cap \mathcal{G}_F}{\mathcal{G}_F}$
    to $\#\subsprob{\mathcal{H}}{\mathcal{G}_F}$ shows $\#\W{1}$-hardness of the latter
    problem by \cref{thm:genhomdich}.
    For the construction of the reduction, we use the known fact that, given a graph $H$,
    there is a quantum graph $Q[H]$ with $\mathsf{supp}(Q)=\mathsf{spasms}(H)$ and
    $\#\homs{Q[H]}{G} = \#\subs{H}{G}$ for every graph $G$~\cite[Equation~(6.2)]{Lovasz12},
    \cite{CurticapeanDM17}.
    Therefore, given a graph $\hat{H} \in \mathsf{spasms}(\mathcal{H})\cap\mathcal{G}_F$
    we can find (in time depending only on $\hat{H}$) a graph $H\in \mathcal{H}$ such that we have
    for all graphs $G$ that
    \[\#\subs{H}{G} = \#\homs{Q[H]}{G},\]
    and $\hat{H}\in \mathsf{supp}(Q[H])$.
    Now let $\hat{Q}[H]$ be the quantum graph obtained from $Q[H]$ by deleting all constituents
    that are not $F$-colorable. If $G$ is $F$-colorable, we obtain from the previous equation that
    \[\#\subs{H}{G} = \#\homs{Q[H]}{G} = \#\homs{\hat{Q}[H]}{G},\]
    as $\#\homs{H'}{G} = 0$ for all graphs $H'$ that are not $F$-colorable:
    Assuming otherwise, there is a homomorphism $h$ from $H'$ to $G$ which yields an
    $F$-coloring of $H'$ when composed with the $F$-coloring of $G$,
    contradicting the assumption that $H'$ is not $F$-colorable.

    It follows that an oracle query $(H,G)$ for $\#\subsprob{\mathcal{H}}{\mathcal{G}_F}$
    computes $\#\homs{\hat{Q}[H]}{G}$.
    As $\mathsf{supp}(\hat{Q}[H])=  \mathsf{spasms}(H) \cap \mathcal{G}_F$, we can use
    \cref{lem:f_col_monotonicity}, which concludes the proof.
\end{proof}

\end{document}